\documentclass[aps, pra,  reprint, superscriptaddress, amsmath, amssymb]{revtex4-2}

\usepackage{appendix}
\usepackage{bbm}
\usepackage{braket}
\usepackage{amsfonts}
\usepackage{amsmath}
\usepackage{amsthm}

\usepackage{txfonts}
\usepackage{bm}
\usepackage[dvipdfmx]{graphicx}
\usepackage{color}
\usepackage{mathrsfs}
\usepackage{subfigure}
\usepackage{wrapfig}
\usepackage{here}
\usepackage{color}
\usepackage{hyperref}

\newcommand{\be}{\begin{equation}}
\newcommand{\ee}{\end{equation}}
\newcommand{\ba}{\begin{align}}
\newcommand{\eaend}{\end{align}}
\newcommand{\bs}{\begin{split}}
\newcommand{\es}{\end{split}}

\theoremstyle{plain}
\newtheorem{thm}{Theorem}
\newtheorem*{thm*}{Theorem}
\newtheorem{lemma}{Lemma}
\theoremstyle{definition}

\begin{document}

\title{Security of passive entanglement-based key distribution protocols}

\author{Shun Kawakami}
\affiliation{Network Innovation Laboratories, NTT, Inc., \\ 1-1 Hikari-no-oka, Yokosuka, Kanagawa 239-0847, Japan}
\affiliation{NTT Research Center for Theoretical Quantum Information, NTT, Inc., \\ 3-1 Morinosato Wakamiya, Atsugi, Kanagawa 243-0198, Japan}

\author{Koichi Takasugi}
\affiliation{Network Innovation Laboratories, NTT, Inc., \\ 1-1 Hikari-no-oka, Yokosuka, Kanagawa 239-0847, Japan}

\author{Koji Azuma}
\affiliation{Basic Research Laboratories, NTT, Inc., \\ 3-1 Morinosato Wakamiya, Atsugi, Kanagawa 243-0198, Japan}
\affiliation{NTT Research Center for Theoretical Quantum Information, NTT, Inc., \\ 3-1 Morinosato Wakamiya, Atsugi, Kanagawa 243-0198, Japan}



\begin{abstract}
Entanglement-based key distribution protocols, such as the Bennett-Brassard-Mermin 1992 (BBM92) protocol and quantum conference key agreement (QCKA), are promising applications of quantum networks. In practical implementations, passive measurement setups are widely adopted because of their simplicity. However, the security analysis of passive protocols with biased basis choice is highly nontrivial, since standard proof techniques for threshold detectors are generally not applicable in this setting.
In this work, we establish the security of passive entanglement-based key distribution protocols in the asymptotic regime. Specifically, we prove the security of passive BBM92 with biased basis choice and extend the proof to passive QCKA with an arbitrary number of parties. In addition, we numerically show that the key generation rate of passive BBM92 is almost identical to that of the corresponding active protocol. Our results provide a theoretical foundation for practical passive implementations of entanglement-based key distribution protocols.

 \end{abstract}

\maketitle

\section{Introduction}

Quantum networks~\cite{2023Azuma} enable distant users to share entanglement, providing a fundamental resource for a wide range of quantum information processing tasks. Among their most promising applications are cryptographic protocols. In the two-party setting, a representative one is quantum key distribution (QKD), whereas its multipartite counterpart is known as quantum conference key agreement (QCKA), where multiple parties establish a common secret key.
One such entanglement-based QKD protocol is Bennett-Brassard-Mermin 1992 (BBM92) protocol~\cite{BBM92}, in which two parties measure distributed Bell states in two complementary bases. Its security has been rigorously established~\cite{1999Lo,  2000Shor, 2008Tsurumaru, 2008Koashi, 2021Lim, 2025Mannalath} and numerous experimental demonstrations~\cite{2000Jennewein, 2000TittelBBM,2007Ursin, 2008Honjo, 2009Erven, 2022Fitzke, 2025Zhuang, 2025Tagliavacche} have confirmed its practical feasibility. Similarly, QCKA protocols based on Greenberger-Horne-Zeilinger (GHZ) states and two measurement bases have been developed, and both their theoretical security~\cite{2018Grasselli} and experimental feasibility~\cite{2021Proietti, 2023Pickston, 2026Zou} have been demonstrated.

In many implementations of these protocols, passive measurement setups with threshold detectors have been adopted because of their simplicity and practicality. In such setups, a beam splitter probabilistically directs incoming signals to measurement devices corresponding to different bases, eliminating the need for active basis switching. Compared with active setups, passive setups reduce hardware complexity and avoid losses, speed limitations, and random-number consumption associated with active basis selection.

The security analysis of passive setups, however, is considerably more involved than that of active setups. This issue is particularly important when the basis choice is biased, either because realistic beam splitters are not perfectly balanced or because biased basis selection can improve the key generation rate by reducing the basis mismatch probability~\cite{2004Lo}. For such scenarios, standard security-proof techniques for threshold detectors are generally not applicable. In particular, the squashing model~\cite{2008Lutkenhaus,2008Tsurumaru}, which reduces the security analysis of realistic optical measurements to that of qubit measurements, is not known for passive measurements with biased basis choice~\cite{2024Kamin}. In addition, security proofs based on complementarity~\cite{2006Koashi,2009Koashi} and the entropic uncertainty principle~\cite{2011Tomamichel} cannot directly be applied when the basis choice is biased~\cite{2025Tupkary}.

Although such a challenge in proving security existed even for the prepare-and-measure Bennett-Brassard 1984 (BB84) protocol~\cite{1984Bennett}, it has recently been overcome, and security has been established even with passive basis choice~\cite{2025Kawakami, 2025Mizutani,2025Wang}. The proof relies on the source-replacement scheme~\cite{2025Tupkary}, where the sender is assumed to possess a virtual qubit whose measurement reproduces the actual state preparation process. This virtual qubit plays a central role in the security analysis.
On the other hand, extending this approach to entanglement-based protocols is highly nontrivial. Unlike BB84, all legitimate parties in BBM92 and QCKA are receivers, and therefore the source-replacement scheme cannot be applied directly. Existing security proofs for entanglement-based protocols typically rely on qubit assumptions that can be justified through a squashing model in active measurement setups. However, such arguments are not available for passive setups with biased basis choice. Consequently, the security of passive entanglement-based protocols has remained an open problem despite their practical importance.

In this paper, we establish the security of passive entanglement-based protocols with biased basis choice in the asymptotic regime. Specifically, we prove the security of the passive BBM92 protocol and its multipartite extension, namely GHZ-state-based QCKA with an arbitrary number of parties. Our key idea is to introduce a virtual-qubit description for entanglement-based protocols that plays a role analogous to that of the source-replacement qubit in BB84. This allows the security proof of passive BBM92 to be reduced to that of BB84 and further enables a security proof for passive QCKA. We additionally present numerical simulations showing that the key rate of passive BBM92 is almost indistinguishable from that of the active protocol, indicating that passive implementations can retain the performance of active schemes while offering substantial practical simplifications.

This paper is organized as follows. In Sec.~\ref{Sec:BBM92}, we introduce the passive BBM92 protocol and present its security proof. We also provide numerical simulations to evaluate its performance. In Sec.~\ref{Sec:QCKA}, we extend the analysis to the passive QCKA protocol and establish its security for an arbitrary number of parties. Finally, Sec.~\ref{Sec:Conclusion} concludes the paper.

\section{Security of passive BBM92 protocol} 
\label{Sec:BBM92}
In this section, we first describe the BBM92 protocol with passive measurement setups and define the observed quantities used in the security analysis. We then establish its security by introducing a virtual-qubit description, which enables the application of a complementarity-based security proof~\cite{2006Koashi,2009Koashi}. Finally, we numerically compare the performance of the passive and active BBM92 protocols.

\begin{figure}[t]
 \centering
 \includegraphics[keepaspectratio, scale=0.5]
      {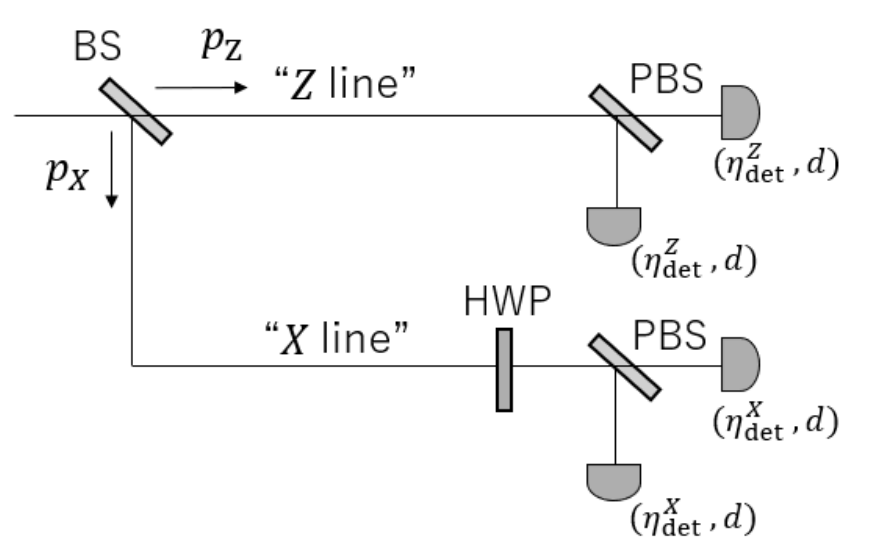}
 \caption{Actual setups for passive basis choice with a beam splitter (BS), polarization beam splitters (PBSs), a half wave plate (HWP) and threshold detectors. 
  Incoming light is split into $Z$ line or $X$ line with the ratio of $p_Z$ to $p_X$. All detectors have the identical dark count probability $d$. The quantum efficiencies of detectors in $Z$ line and $X$ line are $\eta_{\rm det}^Z$ and $\eta_{\rm det}^X$, respectively. The overall transmittances, including the quantum efficiency, in $Z$ line and $X$ line are $\eta_Z$ and $\eta_X$, respectively ($\eta_Z \geq \eta_X$).}
 \label{setup1}
\end{figure}

\subsection{Passive BBM92 Protocol}
\label{Sec: BBMprotocol}
Here, we introduce a polarization-type BBM92 protocol, in which two legitimate parties, Alice and Bob, select their basis passively, while an untrusted party, Fred, is presumed to generate and distribute entangled states. 
In this protocol, two complementary measurement bases ($Z$ and $X$) are used, where the secret key is extracted from the $Z$ basis and the signal disturbance is monitored by using the $X$ basis. The protocol is described as follows. \\

\noindent
(1) {\it State Preparation}: 
Fred is presumed to prepare (not necessarily perfect) a Bell pair and  sends its subsystems to Alice and Bob, respectively. 
\\
(2) {\it Measurement}: 
With the linear optical setup in Fig.~\ref{setup1}, Alice passively splits each incoming signal into the `$Z$ line' and `$X$ line' to perform $Z$-basis and $X$-basis measurements, respectively, with a splitting ratio of $p_Z$:$p_X$ using a beam splitter (BS). In the $X$ line, a half wave plate (HWP), whose transmittance might be smaller than unity, is inserted to rotate the polarization by 45 degrees. In both lines, a polarization beam splitter (PBS) separates the incoming light depending on whether its polarization is horizontal or vertical. 
The four threshold detectors have a dark-count probability $0 \leq d \leq1$ per pulse. 
In the $Z$ ($X$) line, the two detectors have the same quantum efficiency $\eta_{\rm det}^Z$ ($\eta_{\rm det}^X$), and the overall transmittance of the line, including the quantum efficiency, is $\eta_Z$ ($\eta_X$),  satisfying $0 \leq \eta_X \leq \eta_Z \leq 1$ (as $X$ line may have additional optical components such as the HWP). We define a parameter  
\be
r\coloneqq \frac{\eta_X}{\eta_Z}~(\leq 1), \label{rrr}
\ee
representing the asymmetry of transmittance between the two lines. 
Alice determines $W_A \in \{Z, X, \bot, \emptyset \}$ according to the following rule: If none of four detectors clicks, she registers $W_A = \emptyset$. If detections occur only in the $Z$ $(X)$ line, she assigns her measurement basis as $W_A=Z~(X)$. 
If  a `cross click' occurs, corresponding to simultaneous detections across two lines, she sets $W_A = \bot$. 
For events with $W_A=Z~(X)$, if `double click' occurs, meaning that both detectors in the $Z$ $(X)$ line click, she assigns the bit value $b_A \in \{0,1\}$ uniformly at random; otherwise, she determines $b_A$ according to which detector clicks. For events with $W_A = \bot$, no bit is assigned. 
Bob performs the same procedure with the same setup as Alice and obtains  $W_B\in \{Z,X, \bot\, \emptyset \}$ and $b_B \in \{0,1\}$. 
\\
(3) {\it Repetition}: Alice, Bob and Fred repeat steps (1) and (2) $m_{\rm rep}$ times. 
\\
(4) {\it Public communication}: For all $m_{\rm rep}$ rounds, Alice and Bob publicly announce $W_A$ and $W_B$. 
They also disclose $b_A$ for rounds where $W_A=X$ and $b_B$ for rounds where $W_B=X$. 
For rounds where $W_A = W_B = Z$, Alice and Bob randomly select $s=1$ (sampling) with probability $1-\gamma$ and $s=0$ (non-sampling) with probability $\gamma$. 
For rounds with $s=1$, they disclose the corresponding bits $b_A$ and $b_B$. 
\\
(5) {\it Parameter estimation}:
Hereafter, we use the notation $m(\Omega = \omega)$ to represent the number of rounds that satisfy a condition $\Omega = \omega$. 
Alice and Bob obtain a bit string of length 
\be
m_Z \coloneqq m( W_A = W_B = Z, s=0),
\ee
which corresponds to the sifted key. 
For the $Z$-basis sampling rounds, they obtain the numbers of sampled rounds and error rounds,
\be
\begin{split}
m_Z^{\rm sam} &\coloneqq m( W_A = W_B = Z, s=1),\\
m_Z^{\rm sam, err} &\coloneqq m( W_A = W_B = Z, s=1, b_A \neq b_B).
\end{split}
\ee
For the $X$-basis rounds, they obtain the numbers of total rounds and error rounds,
\be
m_X \coloneqq m( W_A = W_B = X), ~~m_{X}^{\rm err}  \coloneqq m(W_A = W_B = X, b_A \neq b_B). 
\ee
Alice and Bob also obtain the numbers of rounds 
\be
m_{\bot_A, Z_B}\coloneqq m(W_A = \bot, W_B = Z),~~m_{Z_A, \bot_B} \coloneqq m( W_A = Z, W_B = \bot),
\ee
corresponding to the events where a cross click occurs on one side and $Z$ basis is chosen on the other side. 
\\
(6) {\it Error correction}: Alice generates a syndrome of length $m_Zf_{\rm EC} $ from her sifted key. 
She encrypts this syndrome using a pre-shared secret key with Bob and sends it to him \cite{2006Koashi}.   
Bob then corrects his key according to the decrypted syndrome. 
\\
(7) {\it Privacy amplification}: Alice and Bob perform privacy amplification by shortening their keys by $m_Z f_{\rm PA} $ bits to obtain the final secret key.
\\

Let us define the following ratios based on the numbers obtained in the above protocol: 
\begin{align}
Q_Z \coloneqq \frac{m_Z}{m_{\rm rep}},~
E_X \coloneqq\frac{m_X^{\rm err}}{m_{\rm rep}}, 
~Q_{\bot_A, Z_B}\coloneqq \frac{m_{\bot_A, Z_B}}{m_{\rm rep}},
~Q_{Z_A, \bot_B}\coloneqq \frac{m_{Z_A, \bot_B}}{m_{\rm rep}},
\label{observed}
\end{align}
and 
\be
e_Z \coloneqq \frac{m_Z^{\rm sam,err}}{m_Z^{\rm sam}}.
\label{def:eZ}
\ee
In this paper, we consider the asymptotic secure key rate in the limit $m_{\rm rep}\to \infty$. 
In this limit, the asymptotic value of $f_{\rm PA}$ is expressed as a function of $Q_Z, E_X , Q_{\bot_A, Z_B}, Q_{Z_A, \bot_B}$. 
The cost for error correction $f_{\rm EC}$ is determined by $e_Z$. 
The key rate $R$ per round is given by
\begin{equation}
R= Q_{Z}\left(1-f_{\rm PA}( Q_Z, E_X , Q_{\bot_A, Z_B}, Q_{Z_A, \bot_B}) -f_{\rm EC}(e_Z)\right).
\label{ki-re-}
\end{equation}

\subsection{Security Proof}
Here, we prove the security of the passive BBM92 protocol described above. We first introduce an alternative protocol that is equivalent to the actual protocol. We then introduce a virtual qubit and define the phase error rate, which determines the amount of privacy amplification required for secure key generation. Finally, we derive an upper bound on the phase error rate and obtain a lower bound on the secure key rate. The main result is summarized in Theorem~\ref{thm:1}.

\begin{figure}[t]
 \centering
 \includegraphics[keepaspectratio, scale=0.5]
      {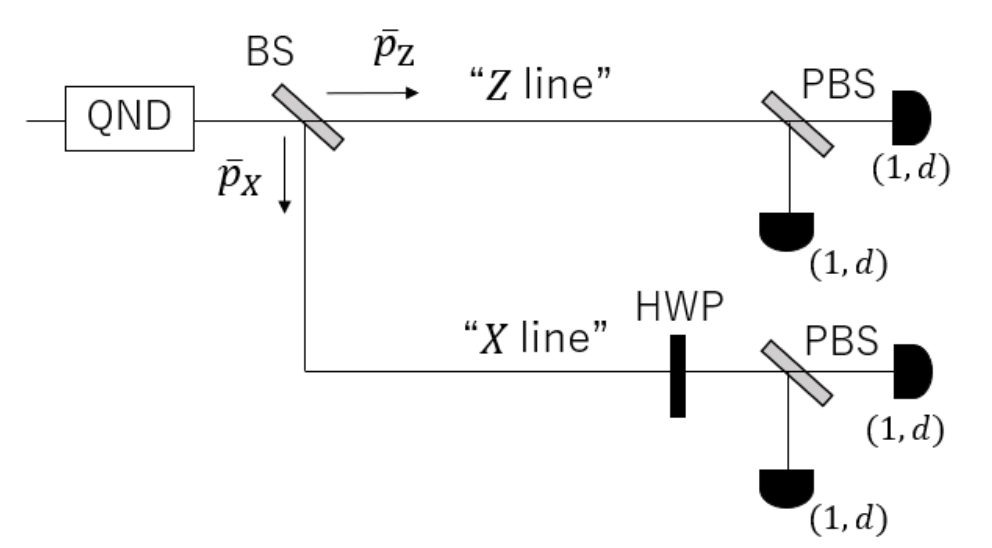}
 \caption{Virtual setup equivalent to Fig.~\ref{setup1}. The QND measurement of the photon number is performed before the first beam splitter, whose transmittance is $\bar{p}_Z = p_Z/(p_Z + p_X r)$. There is no loss in either the $Z$ or $X$ lines, including detector inefficiencies. The detectors have a dark-count probability $d$. }
 \label{setup2}
\end{figure}

\subsubsection{Alternative protocol}
\label{Sec:alternativeBBM}
Our starting point  is to introduce an alternative protocol which is equivalent to the actual protocol from the viewpoint of an eavesdropper, Eve. 
This alternative protocol is defined by replacing step (2) of the actual protocol with the following: 
\\
\\
\noindent
(2') {\it Measurement}: As shown in Fig.~\ref{setup2}, Alice performs a polarization-independent QND measurement to determine the photon number $n_A$. She then passively splits incoming light into  $Z$ line and $X$ line with splitting ratios  
\be
\bar{p}_Z \coloneqq \frac{p_Z}{p_Z + p_X r}, ~~\bar{p}_X \coloneqq \frac{p_X r}{p_Z + p_X r}, \label{barp}
\ee 
respectively. The $Z$ and $X$ lines are assumed to be lossless and equipped with detectors having unit quantum efficiency and a dark count probability $d$ per pulse. 
Note that we have a relation $\bar{p}_Z + \bar{p}_X =1$. 
Bob performs the same operations with the same setup as Alice and obtains the photon number $n_B$. \\
\\
We explain these replacements below. 
The common transmittance $\eta_Z$ shared by $Z$ and $X$ lines in Fig.~\ref{setup1} can be regarded as being absorbed into Eve's quantum channel, while the additional transmittance $r$ (defined in Eq. (\ref{rrr})) for $X$ line remains. 
The setup equivalent to Fig.~\ref{setup1} is therefore depicted in Fig.~\ref{setup3}. 
Now suppose that the loss in $X$ line is modeled by a beam splitter with reflectance $p_X (1-r)$ (i.e., transmittance $p_Z + p_X r$) placed before the splitting into $Z$ and $X$ lines, as shown in Fig.~\ref{setup4}. 
In this case, by choosing the splitting ratio between $Z$ and $X$ lines to be $\bar{p}_Z : \bar{p}_X$, as defined in Eq.~(\ref{barp}), 
the resulting transmittances in the $Z$ and $X$ lines become $(p_Z + p_X r) \bar{p}_Z = p_Z$ and $(p_Z + p_X r) \bar{p}_X = p_X r$, respectively, which are identical to those in Fig.~\ref{setup3}.  
The common transmittance $p_Z + p_X r$ in Fig.~\ref{setup4} can again be absorbed into Eve's channel, so that there is no loss in either the $Z$ or $X$ line. 
Finally, the outcome of the actual measurement remains unchanged even if Alice performs the QND measurement before the beam splitter, because the POVM elements of the actual measurement are block-diagonal with respect to the photon number $n_A$ \cite{2008Lutkenhaus}. 
The same argument applies to Bob. This leads to step (2'). 

\begin{figure}[t]
 \centering
 \includegraphics[keepaspectratio, scale=0.5]
      {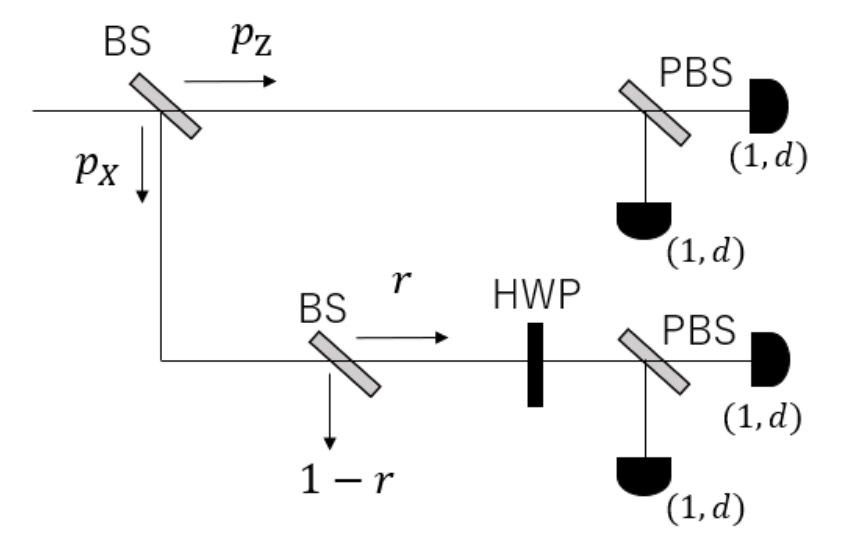}
 \caption{Intermediate virtual setup obtained from Fig.~\ref{setup1} by absorbing the common transmittance $\eta_Z$ into Eve's quantum channel. The additional loss in the $X$ line, characterized by the transmittance $r$, remains unchanged.}
 \label{setup3}
\end{figure}

\begin{figure}[t]
 \centering
 \includegraphics[keepaspectratio, scale=0.5]
      {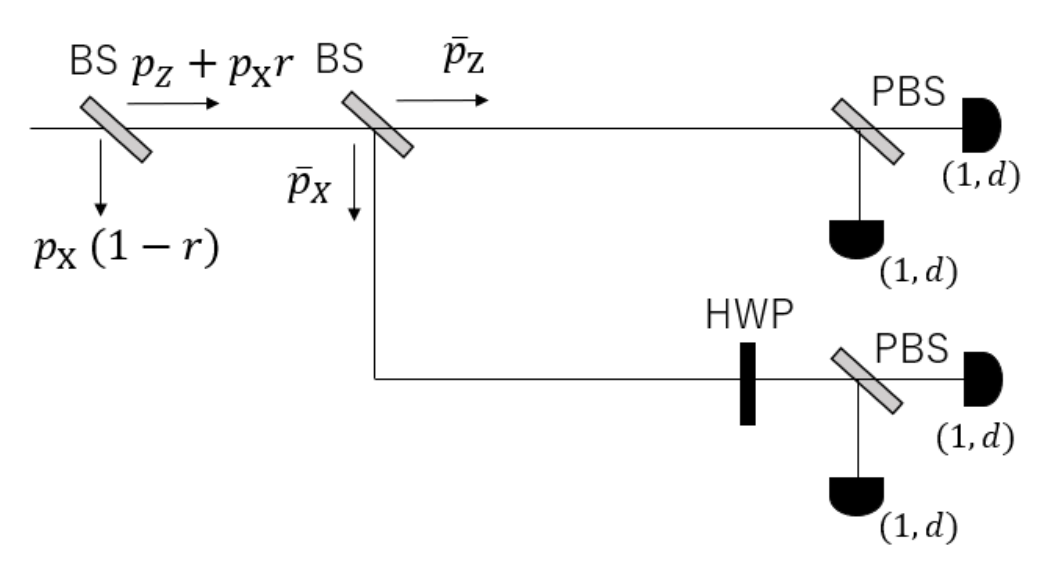}
 \caption{Intermediate virtual setup in the derivation of Fig.~\ref{setup2}, equivalent to Fig.~\ref{setup3}. The additional loss in the $X$ line is modeled by a beam splitter with transmittance $p_Z+p_Xr$ placed before the splitting into the $Z$ and $X$ lines. The splitting ratios are chosen to be $\bar{p}_Z:\bar{p}_X$, yielding the same transmittances in the two lines as in Fig.~\ref{setup3}.}
 \label{setup4}
\end{figure}

In the alternative protocol, photon numbers $n_A$ and $n_B$ are treated as observed quantities, 
which allows us to define the following ratios: 
\begin{equation}
\begin{split}
Q_{Z}^{(n'_A, n'_B)} &\coloneqq m(W_A=W_B=Z, s=0, n_A=n'_A,  n_B=n'_B) / m_{\rm rep}, \\
E_{X}^{(n'_A, n'_B)} &\coloneqq m(W_A=W_B=X,  n_A=n'_A, n_B=n'_B, b_A \neq b_B) / m_{\rm rep}, \\
Q_{\bot_A, Z_B}^{(n'_A, n'_B)} &\coloneqq m(W_A=\bot, W_B=Z, n_A=n'_A, n_B=n'_B) / m_{\rm rep}, \\
Q_{Z_A, \bot_B}^{(n'_A, n'_B)} &\coloneqq m(W_A=Z, W_B=\bot, n_A=n'_A, n_B=n'_B) / m_{\rm rep}.
\label{alternativeQE}
\end{split}
\end{equation}

\subsubsection{Phase error and  privacy amplification} \label{Sec:phasePA}
In this paper, we adopt the security proof with complementarity, in which the `phase error' rate is bounded by observed quantities in the actual protocol~\cite{2006Koashi, 2009Koashi}.
To consider the amount of privacy amplification for the sifted key obtained when $W_A=W_B=Z$, we need to define a virtual qubit on system $A$ such that bits obtained by performing the $Z$-basis measurement on this qubit 
are equivalent to the sifted key~\cite{2009Koashi}. 
When $n_A \neq 1$, Alice generates the virtual qubit in the following rule: she first performs the actual measurement with the setup in Fig.~\ref{setup2}. If she obtains an outcome for the case of $W_A=Z$ and $b_A=0~(1)$, she prepares a single photon in the $Z$-basis state $\ket{H}_A~(\ket{V}_A)$ as the virtual qubit on system $A$. Here, $\ket{H}_A$ and $\ket{V}_A$ represent single-photon states with horizontal and vertical polarizations, respectively. 
On the other hand, when $n_A = 1$, the rule to obtain the virtual qubit will be shown in the proof of lemma \ref{lemmasingle} in Sec.~\ref{lemmaproof}.  
Once the virtual qubit on system $A$ is defined for arbitrary $n_A$, in the security proof with complementarity, the `phase error' is defined in the following virtual scenario: suppose that, for the events where $W_A=W_B=Z$, the virtual qubit is measured in the $X$-basis measurement $\{\ket{D}\bra{D}_A,  \ket{\bar{D}} \bra{\bar{D}}_A \}$ with $\ket{D}_A \coloneqq (\ket{H}_A + \ket{V}_A) / \sqrt{2}$ and $\ket{\bar{D}}_A \coloneqq (\ket{H}_A - \ket{V}_A) / \sqrt{2}$ to obtain the outcome $\tilde{b}_A$ instead of the actual $Z$-basis measurement. Bob then tries to guess the outcome by using his quantum state. 
By defining $\tilde{b}_B$ as a Bob's guess bit, the `phase error' is then defined as a bit error between $\tilde{b}_A$ and $\tilde{b}_B$. 
Since $X$ and $Z$ bases are complementary, the better Bob can predict the outcomes of Alice's $X$-basis measurement, i.e., the smaller the phase error rate, the less information about Alice's $Z$-basis measurement outcomes (which form the secret key) is leaked to Eve. In other words, a smaller phase error rate implies a smaller amount of privacy amplification. 
For the events with $n_A \neq 1$, the phase error rate is 1/2 because $X$-basis measurement is performed on the prepared $Z$-basis state $\ket{H}_A$ or $\ket{V}_A$ and Bob cannot guess the outcome at all.  
In addition, we assume the worst case where Bob gives up his guess for the outcome when $n_B \neq 1$, corresponding to a phase error rate of 1/2 in those cases. 
These imply that in our security proof, the secret key is extracted solely from the events with $n_A = n_B = 1$, while all other events are discarded through privacy amplification. 
For such events with $n_A = n_B = 1$, we define 
\begin{align}
E_{\rm ph}^{(1, 1)}& \coloneqq m(W_A=W_B=Z, s=0, n_A=n_B=1, \tilde{b}_A \neq \tilde{b}_B) / m_{\rm rep}.  
\label{phaseerrorratio}
\end{align}
Then the phase error rate is represented as  $E_{\rm ph}^{(1, 1)} / Q_Z^{(1, 1)}$. 
Since all events except those with $n_A = n_B = 1$ have the phase error rate 1/2, the fraction of bits to be sacrificed in privacy amplification is given by 
\begin{align}
\begin{split}
&Q_Z f_{\rm PA} = \\
&~~\sum_{\substack{n_A, n_B \\ n_A \ge 2 \lor n_B \ge 2}} Q_Z^{(n_A, n_B)} + Q_Z^{(0,0)} + Q_Z^{(0,1)} + Q_Z^{(1,0)} +Q_Z^{(1,1)} h\left( \frac{E_{\rm ph}^{(1,1)}}{Q_Z^{(1,1)}}\right), \label{PAamount}
\end{split}
\end{align}
where $h(x)$ is the binary entropy function defined as 
\be 
h(x) \coloneqq
-x {\rm log}_2 x - (1-x) {\rm log}_2(1-x). 
\ee

\subsubsection{Main theorem} \label{lemmaproof}
From Eq.~(\ref{ki-re-}), it suffices to derive an upper bound on $f_{\rm PA}$ in order to obtain the secure key rate from the observed quantities $Q_Z$, $E_X$, $Q_{\bot_A,Z_B}$ and $Q_{Z_A,\bot_B}$.
The result is summarized in Theorem~\ref{thm:1}.
The proof of the theorem is based on three lemmas.
Lemma~\ref{lemmasingle} provides an upper bound on the phase error rate appearing in the last term of Eq.~(\ref{PAamount}).
Lemma~\ref{lemmamulti} provides an upper bound on the multi-photon contribution, corresponding to the first term of Eq.~(\ref{PAamount}).
Lemma~\ref{lemmaconcavity} provides an upper bound on the sum of the second through fifth terms of Eq.~(\ref{PAamount}) by exploiting the concavity of the binary entropy function.

We begin with the following lemma, which provides an upper bound on the phase error ratio $E_{\rm ph}^{(1,1)}$ for single-photon events.

\begin{lemma}\label{lemmasingle}
For events where $n_A = n_B = 1$, the following relation holds:
\be
E^{(1,1)}_{\rm ph} \sim \frac{{\bar{p}_Z}^2 \gamma}{{\bar{p}_X}^2} E^{(1,1)}_X,
\ee
where '$\sim$' denotes asymptotic equality in the limit $m_{\rm rep} \to \infty$. 
\end{lemma}

\begin{proof}
Suppose that system $A$ is the input system to the beam splitter in Fig.~\ref{setup2}.
As shown in Appendix \ref{singlePOVM}, the POVM elements in system $A$, conditioned on $n_A = 1$, are 
\begin{equation}
\begin{split}
&\hat{F'}_{Z_A0}^{(1)} =  \bar{p}_Z  \ket{H} \bra{H}_A (1-d)^3, \\
&\hat{F'}_{Z_A1}^{(1)} =  \bar{p}_Z  \ket{V} \bra{V}_A (1-d)^3,  \\
&\hat{F'}_{Z_A, \rm double}^{(1)} = \bar{p}_Z d  (1-d)^2 \hat{\mathbbm{1}}_A^{(1)},  \\
&\hat{F'}_{X_A0}^{(1)} =  \bar{p}_X  \ket{D} \bra{D}_A (1-d)^3, \\
&\hat{F'}_{X_A1}^{(1)} =  \bar{p}_X  \ket{\bar{D}} \bra{\bar{D}}_A  (1-d)^3,  \\
&\hat{F'}_{X_A, \rm double}^{(1)} = \bar{p}_X d  (1-d)^2 \hat{\mathbbm{1}}_A^{(1)},  \\
&\hat{F'}_{\bot_A}^{(1)} = \left( 1-(1-d)^2 \right) \hat{\mathbbm{1}}_A^{(1)},
\end{split}
\end{equation}
where the subscripts indicate the corresponding measurement outcomes and $ \hat{\mathbbm{1}}_A^{(1)} \coloneqq \ket{H}\bra{H}_A + \ket{V}\bra{V}_A$. 
Note that the sum of those elements equals $ \hat{\mathbbm{1}}_A^{(1)}$. 
Including the classical post-processing, the possible outcomes are restricted to `0', `1' and `failure', where `failure' corresponds to a cross-click event. 
Since a random bit is assigned to a double-click event, POVM elements including the post-processing are given by
\begin{equation}
\begin{split}
&\hat{G}_{Z_A0}^{(1)} = \hat{F'}_{Z_A0}^{(1)} + \frac{1}{2} \hat{F'}_{Z_A, \rm double}^{(1)}, ~~
\hat{G}_{Z_A1}^{(1)} = \hat{F'}_{Z_A1}^{(1)} + \frac{1}{2} \hat{F'}_{Z_A, \rm double}^{(1)}, \\
&\hat{G}_{X_A0}^{(1)} = \hat{F'}_{X_A0}^{(1)} + \frac{1}{2} \hat{F'}_{X_A, \rm double}^{(1)}, ~~
\hat{G}_{X_A1}^{(1)} = \hat{F'}_{X_A1}^{(1)} + \frac{1}{2} \hat{F'}_{X_A, \rm double}^{(1)}, \\
&\hat{G}_{{\rm fail}_A}^{(1)} =  \hat{F'}_{\bot_A}^{(1)} . 
\label{POVMs}
\end{split}
\end{equation}
Since $\hat{G}_{{\rm fail}_A}^{(1)} =  \hat{F'}_{\bot_A}^{(1)} = (1-(1-d)^2) \hat{\mathbbm{1}}_A^{(1)}$ holds, the measurement is equivalent to a process in which 
any single-photon state is first filtered with success probability $(1-d)^2$, followed by the measurement $\{\hat{G}_{Z_A0}^{(1)}/(1-d)^2$, $\hat{G}_{Z_A1}^{(1)}/(1-d)^2$, $\hat{G}_{X_A0}^{(1)}/(1-d)^2$, $\hat{G}_{X_A1}^{(1)}/(1-d)^2\}$.
Here, each element is written by
\be
\begin{split}
&\frac{\hat{G}_{Z_A0}^{(1)}}{(1-d)^2} = \bar{p}_Z \left((1-d) \ket{H}\bra{H}_A +\frac{d}{2} \hat{\mathbbm{1}}_A^{(1)} \right), \\
&\frac{\hat{G}_{Z_A1}^{(1)}}{(1-d)^2} = \bar{p}_Z \left((1-d) \ket{V}\bra{V}_A +\frac{d}{2} \hat{\mathbbm{1}}_A^{(1)} \right),\\ 
&\frac{\hat{G}_{X_A0}^{(1)}}{(1-d)^2} = \bar{p}_X \left((1-d) \ket{D}\bra{D}_A +\frac{d}{2} \hat{\mathbbm{1}}_A^{(1)} \right), \\
&\frac{\hat{G}_{X_A1}^{(1)}}{(1-d)^2} = \bar{p}_X \left((1-d) \ket{\bar{D}}\bra{\bar{D}}_A +\frac{d}{2} \hat{\mathbbm{1}}_A^{(1)} \right). \\ 
\end{split}
\ee
From these expressions, we obtain the following equations:
\be
\begin{split}
&\frac{\hat{G}_{Z_A0}^{(1)}}{(1-d)^2} =  \sum\limits_{l=0}^3 \bar{p}_Z \hat{K}_{A,l}^{\dagger} \ket{H}\bra{H}_A \hat{K}_{A,l}, \\
&\frac{\hat{G}_{Z_A1}^{(1)}}{(1-d)^2} =  \sum\limits_{l=0}^3 \bar{p}_Z \hat{K}_{A,l}^{\dagger} \ket{V}\bra{V}_A \hat{K}_{A,l}, \\
&\frac{\hat{G}_{X_A0}^{(1)}}{(1-d)^2} =  \sum\limits_{l=0}^3 \bar{p}_X \hat{K}_{A,l}^{\dagger} \ket{D}\bra{D}_A \hat{K}_{A,l}, \\
&\frac{\hat{G}_{Z_A1}^{(1)}}{(1-d)^2} =  \sum\limits_{l=0}^3 \bar{p}_X \hat{K}_{A,l}^{\dagger} \ket{\bar{D}}\bra{\bar{D}}_A \hat{K}_{A,l}, \\
\end{split}
\ee
where $\{\hat{K}_{A,l}\}_{l=0,1,2,3}$ are Kraus operators on system $A$ defined as 
\be
\begin{split}
\hat{K}_{A,0} &\coloneqq \sqrt{1-3d/4}\hat{\mathbbm{1}}_A^{(1)}, ~\hat{K}_{A,1} \coloneqq \sqrt{d/4}\hat{\sigma}_{A,X}, \\
\hat{K}_{A,2} &\coloneqq \sqrt{d/4} \hat{\sigma}_{A,Y}, ~\hat{K}_{A,3} \coloneqq \sqrt{d/4} \hat{\sigma}_{A,Z},
\label{Krauses}
\end{split}
\ee
with $\{\hat{\sigma}_{A,W}\}_{W=X,Y,Z}$ being Pauli operators on system $A$ defined as $\hat{\sigma}_{A,X} \coloneqq \ket{H}\bra{V}_A + \ket{V}\bra{H}_A$, $\hat{\sigma}_{A,Y} \coloneqq -i\ket{H}\bra{V}_A + i\ket{V}\bra{H}_A$ and $\hat{\sigma}_{A,Z} \coloneqq \ket{H}\bra{H}_A - \ket{V}\bra{V}_A$. 
Therefore, Alice's actual measurement for events with $n_A = 1$ is equivalent to a process where the incoming state is first filtered with success probability $(1-d)^2$, then the CPTP map with Kraus operators $\{\hat{K}_{A,l}\}_{l=0,1,2,3}$ is applied, followed by active basis selection with probability $\bar{p}_Z$ and $\bar{p}_X$. Finally, the ideal measurement $\{\ket{H}\bra{H}_A, \ket{V}\bra{V}_A\}$ and $\{\ket{D}\bra{D}_A, \ket{\bar{D}}\bra{\bar{D}}_A\}$ are performed for the $Z$ and $X$ bases, respectively.

In Bob's system $B$, which is the input system to the beam splitter in Fig.~\ref{setup2}, we define  $\ket{H}_B$, $\ket{V}_B$, $\ket{D}_B$ and $\ket{\bar{D}}_B$ as single-photon states with horizontal, vertical, diagonal and anti-diagonal polarizations, respectively. 
Then, the same discussion above for Alice holds for Bob. Namely, for events with $n_B=1$, Bob's actual measurement is equivalent to the process where the filtering operation with success probability $(1-d)^2$ and CPTP map specified by $\{\hat{K}_{B,l}\}_{l=0,1,2,3}$ in Eq.~(\ref{Krauses}) are applied to the received single photon, followed by the measurement $\{\ket{H}\bra{H}_B, \ket{V}\bra{V}_B\}$ or $\{\ket{D}\bra{D}_B, \ket{\bar{D}}\bra{\bar{D}}_B\}$ according to his active basis choice $Z$ or $X$, respectively.

Now, as announced in Sec.~\ref{Sec:phasePA}, we define Alice's virtual qubit on system $A$ for the events with $n_A=1$ as follows: Alice first applies the filter and the CPTP map specified by $\{\hat{K}_{A,l}\}$ to her received single photon, and keeps the output as the virtual qubit. 
For the events where Alice and Bob select $Z$ basis and the event is chosen for non-sampling ($W_A=W_B=Z, s=0$), the $Z$-basis measurement $\{\ket{H}\bra{H}_A,  \ket{V}\bra{V}_A\}$ on Alice's virtual qubit is equivalent to the actual process of obtaining a sifted key. 
On the other hand, for such events, the phase error is defined as a bit error between the outcome $\tilde{b}_A$ of Alice's $X$-basis measurement $\{\ket{D}\bra{D}_A, \ket{\bar{D}}\bra{\bar{D}}_A\}$ on her virtual qubit and Bob's guess $\tilde{b}_B$. 
Here, we define Bob's guess $\tilde{b}_B$ for events with $n_B=1$ as the outcome of an $X$-basis measurement performed on the single photon. 
Accordingly, the phase error for events with $n_A=n_B=1$ is described as follows. Alice and Bob first independently apply a filter with success probability $(1-d)^2$, followed by the CPTP maps $\{\hat{K}_{A,l}\}_{l=0,1,2,3}$ and $\{\hat{K}_{B,l}\}_{l=0,1,2,3}$. They then independently choose $Z$ basis with probability $\bar{p}_Z$, and select the events as non-sampling with probability $\gamma$. 
For the resulting events, the phase error is defined as the discrepancy between the outcomes $\tilde{b}_A$ and $\tilde{b}_B$ of $X$-basis measurements on the single photons. 
Then, the POVM element $\hat{E}_{\rm ph}^{(1,1)} $ corresponding to obtain the phase error event with ($W_A = W_B = Z, s= 0, n_A= n_B=1, \tilde{b}_A \neq \tilde{b}_B$) is represented as 
\be
\hat{E}_{\rm ph}^{(1,1)}  = \gamma \left( \hat{G}_{{\rm ph}_A0}^{(1)} \otimes \hat{G}_{{\rm ph}_B1}^{(1)} + \hat{G}_{{\rm ph}_A1}^{(1)} \otimes \hat{G}_{{\rm ph}_B0}^{(1)} \right ),
\label{POVMEph}
\ee
where 
\be
\begin{split}
&\hat{G}_{{\rm ph}_A0}^{(1)} \coloneqq  \bar{p}_Z (1-d)^2 \sum\limits_{l=0}^3  \hat{K}_{A,l}^{\dagger} \ket{D}\bra{D}_A \hat{K}_{A,l}, \\
&\hat{G}_{{\rm ph}_A1}^{(1)} \coloneqq  \bar{p}_Z (1-d)^2 \sum\limits_{l=0}^3  \hat{K}_{A,l}^{\dagger} \ket{\bar{D}}\bra{\bar{D}}_A \hat{K}_{A,l}, \\
&\hat{G}_{{\rm ph}_B0}^{(1)} \coloneqq  \bar{p}_Z (1-d)^2 \sum\limits_{l=0}^3  \hat{K}_{B,l}^{\dagger} \ket{D}\bra{D}_B \hat{K}_{B,l}, \\
&\hat{G}_{{\rm ph}_B1}^{(1)} \coloneqq  \bar{p}_Z (1-d)^2 \sum\limits_{l=0}^3  \hat{K}_{B,l}^{\dagger} \ket{\bar{D}}\bra{\bar{D}}_B \hat{K}_{B,l}.
\end{split}
\ee
On the other hand, the POVM element $\hat{E}_X^{(1,1)} $ corresponding to the $X$-basis error event with ($W_A = W_B = X, n_A= n_B=1, b_A \neq b_B$) is represented as  
\be
\hat{E}_X^{(1,1)} \coloneqq G_{X_A0}^{(1)} \otimes G_{X_B1}^{(1)} + G_{X_A1}^{(1)} \otimes G_{X_B0}^{(1)}, 
\label{POVMEX}
\ee
where we defined 
\be
\begin{split}
&\hat{G}_{X_B0}^{(1)} \coloneqq  \bar{p}_X (1-d)^2 \sum\limits_{l=0}^3  \hat{K}_{B,l}^{\dagger} \ket{D}\bra{D}_B \hat{K}_{B,l}, \\
&\hat{G}_{X_B1}^{(1)} \coloneqq  \bar{p}_X (1-d)^2 \sum\limits_{l=0}^3  \hat{K}_{B,l}^{\dagger} \ket{\bar{D}}\bra{\bar{D}}_B \hat{K}_{B,l}.
\end{split}
\ee
Eqs.~(\ref{POVMEph}) and (\ref{POVMEX}) lead to the relation
\be
\hat{E}^{(1,1)}_{\rm ph} = \frac{{\bar{p}_Z}^2 \gamma}{{\bar{p}_X}^2} \hat{E}^{(1,1)}_X.
\ee
Let $\hat{F}(\Omega = \omega)$ denote a POVM element corresponding to the outcome of $\Omega = \omega$.
From Azuma's inequality~\cite{1967Azuma} (see Appendix \ref{Azuma}), 
if 
\be
\hat{F}(\Omega = \omega) = a \hat{F}(\Omega = \omega')
\ee
holds for a constant $a>0$, then 
\be
\frac{m(\Omega = \omega)}{m_{\rm rep}} = \frac{a m(\Omega = \omega')}{m_{\rm rep}}
\ee
holds in the asymptotic limit.  
Therefore, 
we obtain
\be
E^{(1,1)}_{\rm ph} \sim \frac{{\bar{p}_Z}^2 \gamma}{{\bar{p}_X}^2} E^{(1,1)}_X
\ee
in the asymptotic limit. 

\end{proof}

Here, we define 
\be
\alpha \coloneqq \frac{\bar{p}_Z^2 \gamma  }
  {1- \bar{p}_Z^2 - \bar{p}_X^2 }.
\label{alpha}
\ee
The following lemma gives an upper bound on the contribution from multi-photon events. 
\begin{lemma}\label{lemmamulti}
The following relation holds:  
\be
\sum_{\substack{n_A, n_B \\ n_A \ge 2 \lor n_B \ge 2}} Q_Z^{(n_A, n_B)} \lesssim \alpha( Q_{\bot_A, Z_B} + Q_{Z_A, \bot_B}),
\label{templemmamulti}
\ee
where '$\lesssim$' denotes inequality that holds in the asymptotic limit. 
\end{lemma}

\begin{proof}  The left-hand side in Eq.~(\ref{templemmamulti}) can be upper bounded as 
\be
\sum_{\substack{n_A, n_B \\ n_A \ge 2 \lor n_B \ge 2}} Q_Z^{(n_A, n_B)} \leq Q_Z^{(n_A\geq 2)}  + Q_Z^{(n_B\geq 2)},
\label{nAnBupper}
\ee
where we defined 
\be
Q_Z^{(n_A\geq 2)} \coloneqq \sum_{\substack{n_A, n_B \\ n_A \ge 2 \land n_B \ge 0}} Q_Z^{(n_A, n_B)}, ~~
Q_Z^{(n_B\geq 2)} \coloneqq \sum_{\substack{n_A, n_B \\ n_A \ge 0 \land n_B \ge 2}} Q_Z^{(n_A, n_B)}.
\ee
From Appendix \ref{multiPOVM}, for $n_A \geq 1$, the POVM elements corresponding to the outcomes $W_A = Z$, $W_A = X$ and $W_A = \bot$ are given by 
\begin{align}
\bs
&\hat{F'}_{Z_A}^{(n_A)} = (1-d)^2 {\bar{p}_Z}^{n_A} \hat{\mathbbm{1}}_A^{(n_A)},~~~\hat{F'}_{X_A}^{(n_A)} = (1-d)^2 {\bar{p}_X}^{n_A} \hat{\mathbbm{1}}_A^{(n_A)}, \\
&\hat{F'}_{\bot_A}^{(n_A)} = \left( 1 - (1-d)^2 (\bar{p}_Z^{n_A} + \bar{p}_X^{n_A}) \right) \hat{\mathbbm{1}}_A^{(n_A)},
\label{FdashA}
\es
\end{align}
respectively, where $\hat{\mathbbm{1}}_A^{(n_A)}$ denotes the projection operator on the $n_A$-photon subspace of system $A$. 
For $n_A = 0$, the POVM elements are given by
\be
\begin{split}
\hat{F'}_{Z_A}^{(0)} &= \hat{F'}_{X_A}^{(0)} = (1-d)^2 (1-(1-d)^2) \ket{\rm vac}\bra{\rm vac}_A,\\
\hat{F'}_{\bot_A}^{(0)} &=  (1 - (1-d)^2)^2  \ket{\rm vac}\bra{\rm vac}_A, 
\label{FdashA0}
\end{split}
\ee
where $\ket{\rm vac}_A$ is the vacuum state on system $A$. 
On Bob's side, $\hat{F'}_{Z_B}^{(n_B)}$, $\hat{F'}_{X_B}^{(n_B)}$, $\hat{F'}_{\bot_B}^{(n_B)}$ are defined in the same manner as 
\begin{align}
\bs
&\hat{F'}_{Z_B}^{(n_B)} = (1-d)^2 {\bar{p}_Z}^{n_B} \hat{\mathbbm{1}}_B^{(n_B)},~~~\hat{F'}_{X_B}^{(n_B)} = (1-d)^2 {\bar{p}_X}^{n_B} \hat{\mathbbm{1}}_B^{(n_B)}, \\
&\hat{F'}_{\bot_B}^{(n_B)} = \left( 1 - (1-d)^2 (\bar{p}_Z^{n_B} + \bar{p}_X^{n_B}) \right) \hat{\mathbbm{1}}_B^{(n_B)}
\label{FdashB}
\es
\end{align} 
for $n_B \geq 1$, and 
\be
\begin{split}
\hat{F'}_{Z_B}^{(0)} &= \hat{F'}_{X_B}^{(0)} = (1-d)^2 (1-(1-d)^2) \ket{\rm vac}\bra{\rm vac}_B,\\
\hat{F'}_{\bot_B}^{(0)} &=  (1 - (1-d)^2)^2  \ket{\rm vac}\bra{\rm vac}_B 
\label{FdashB0}
\end{split}
\ee
for $n_B = 0$. Here $\ket{\rm vac}_B$ denotes the vacuum state on system $B$. 

Here, we define the following operators for $n_A \geq 0$ and $n_B \geq 0$ corresponding to the sifted-key and cross-click events:
\be
\begin{split}
\hat{Q}_Z^{(n_A, n_B)} &\coloneqq \gamma \hat{F'}_{Z_A}^{(n_A)} \otimes \hat{F'}_{Z_B}^{(n_B)}, \\
\hat{Q}_{\bot_A, Z_B}^{(n_A, n_B)}  &\coloneqq \hat{F'}_{\bot_A}^{(n_A)} \otimes \hat{F'}_{Z_B}^{(n_B)}, \\
\hat{Q}_{Z_A, \bot_B}^{(n_A, n_B)}  &\coloneqq \hat{F'}_{Z_A}^{(n_A)} \otimes \hat{F'}_{\bot_B}^{(n_B)}, \\
\label{hatQn}
\end{split}
\ee
where $\gamma$ is the non-sampling probability. 
By using the expressions of Eq.~(\ref{FdashA}), we obtain 
\be
\begin{split}
\hat{Q}_Z^{(n_A, n_B)} & =  \frac{ \bar{p}_Z^{n_A} (1-d)^2 \gamma  }
{1 - (\bar{p}_Z^{n_A} + \bar{p}_X^{n_A})(1-d)^2}  \hat{Q}_{\bot_A, Z_B}^{(n_A, n_B)} \\
& \leq \frac{ \bar{p}_Z^{n_A} \gamma  }
{1 - (\bar{p}_Z^{n_A} + \bar{p}_X^{n_A})}  \hat{Q}_{\bot_A, Z_B}^{(n_A, n_B)} 
\label{operatorrelation}
\end{split}
\ee
for $n_A \geq 1$ and $n_B \geq 0$. 
For $n_A \geq 2$ and $n_B \geq 0$, the relation 
\be
\hat{Q}_Z^{(n_A, n_B)} \leq \alpha \hat{Q}_{\bot_A, Z_B}^{(n_A, n_B)} 
\ee
is obtained, where $\alpha$ is defined in Eq.~(\ref{alpha}). 
Therefore, 
\be
\sum_{\substack{n_A, n_B \\ n_A \ge 2 \land n_B \ge 0}} \hat{Q}_Z^{(n_A, n_B)} 
\leq \sum_{\substack{n_A, n_B \\ n_A \ge 2 \land n_B \ge 0}} \alpha \hat{Q}_{\bot_A, Z_B}^{(n_A, n_B)}
\leq \sum_{\substack{n_A, n_B \\ n_A \ge 0 \land n_B \ge 0}} \alpha \hat{Q}_{\bot_A, Z_B}^{(n_A, n_B)}
\label{QZnAupper}
\ee
holds. 
By defining 
\be
\hat{Q}_Z^{(n_A \geq 2)} \coloneqq \sum_{\substack{n_A, n_B \\ n_A \ge 2 \land n_B \ge 0}} \hat{Q}_Z^{(n_A, n_B)} ,
~~ \hat{Q}_{\bot_A, Z_B} \coloneqq \sum_{\substack{n_A, n_B \\ n_A \ge 0 \land n_B \ge 0}} \alpha \hat{Q}_{\bot_A, Z_B}^{(n_A, n_B)},
\ee
Eq.~(\ref{QZnAupper}) is 
\be
\hat{Q}_Z^{(n_A \geq 2)} \leq \alpha \hat{Q}_{\bot_A, Z_B}.
\label{operela}
\ee
By applying Azuma's inequality (see Appendix \ref{Azuma}), we obtain 
\be
Q_Z^{(n_A \geq 2)} \lesssim \alpha Q_{\bot_A, Z_B}
\label{QZnAupperlast}
\ee
in the asymptotic limit.

Similarly, we have 
\be
Q_Z^{(n_B \geq 2)} \lesssim  \alpha Q_{Z_A, \bot_B}
\label{QZnBupper}
\ee
in the asymptotic limit. 
From Eqs.~(\ref{nAnBupper}), (\ref{QZnAupperlast}) and (\ref{QZnBupper}), we obtain 
\be
\sum_{\substack{n_A, n_B \\ n_A \ge 2 \lor n_B \ge 2}} Q_Z^{(n_A, n_B)} \leq \alpha( Q_{\bot_A, Z_B} + Q_{Z_A, \bot_B}).
\ee
\end{proof}

Now, we define 
\be
q_Z \coloneqq Q_Z^{(0,0)} + Q_Z^{(0,1)} + Q_Z^{(1,0)} +Q_Z^{(1,1)} = Q_Z - \sum_{\substack{n_A, n_B \\ n_A \ge 2 \lor n_B \ge 2}} Q_Z^{(n_A, n_B)}
\label{qz}
\ee
for convenience. 
The following lemma provides an upper bound on the contribution to $Q_Z f_{\rm PA}$ other than the multi-photon term, by exploiting the concavity of the binary entropy function. 
\begin{lemma}\label{lemmaconcavity}
Suppose $\bar{p}_Z \geq \bar{p}_X$ and 
\be
\frac{\bar{p}_Z^2 \gamma}{\bar{p}_X^2} \frac{E_X}{q_Z} \leq 1/2.
\ee
Then the following inequality holds:
\be
Q_Z^{(0,0)} + Q_Z^{(0,1)} + Q_Z^{(1,0)} +Q_Z^{(1,1)} h\left(\frac{\bar{p}_Z^2 \gamma}{\bar{p}_X^2} \frac{E_X^{(1,1)}}{Q_Z^{(1,1)}}\right) \lesssim q_Z h \left( \frac{\bar{p}_Z^2 \gamma}{\bar{p}_X^2} \frac{E_X}{q_Z} \right). 
\label{eqlemma2}
\ee
\end{lemma}

\begin{proof}
The quantity $E_X^{(1, 1)}$ can be upper bounded as follows:
\begin{align}
\bs
E_X^{(1, 1)} &= E_X - E_X^{(0, 0)} - E_X^{(0, 1)} -E_X^{(1, 0)} - \sum_{\substack{n_A, n_B \\ n_A \ge 2 \lor n_B \ge 2}} E_X^{(n_A, n_B)} \\
& \leq E_X - E_X^{(0, 0)} - E_X^{(0, 1)} -E_X^{(1, 0)}.
\es
\label{Xerrorineq}
\end{align}
Let $\hat{E}_X^{(0, 0)}, \hat{E}_X^{(0, 1)}, \hat{E}_X^{(1, 0)}$ be POVM elements corresponding to the events with $(W_A = W_B = X, n_A=n_B=0, b_A \neq b_B)$, $(W_A = W_B = X, n_A=0, n_B=1, b_A \neq b_B)$ and $(W_A = W_B = X, n_A=1,n_B=0, b_A \neq b_B)$, respectively. These are given by 
\be
\begin{split}
\hat{E}_X^{(0, 0)} &= \frac{1}{2}(1-d)^4 \left(1-(1-d)^2\right)^2 \ket{\rm vac} \bra{\rm vac}_{AB}, \\
\hat{E}_X^{(0, 1)} &= \frac{\bar{p}_X}{2} (1-d)^4 \left(1-(1-d)^2\right)  \ket{\rm vac} \bra{\rm vac}_A \otimes \hat{\mathbbm{1}}_B^{(1)},  \\ 
\hat{E}_X^{(1, 0)} &= \frac{\bar{p}_X}{2} (1-d)^4 \left(1-(1-d)^2\right) \hat{\mathbbm{1}}_A^{(1)} \otimes \ket{\rm vac} \bra{\rm vac}_B,
\end{split}
\ee
where the prefactor 1/2 implies the randomness of outcomes from dark counts. 
From Eq.~(\ref{hatQn}), we obtain 
\be
\hat{E}_X^{(0, 0)} = \frac{1}{2 \gamma}  \hat{Q}_Z^{(0, 0)},~~\hat{E}_X^{(0, 1)} = \frac{1}{2 \gamma} \frac{\bar{p}_X}{\bar{p}_Z} \hat{Q}_Z^{(0, 1)},~~\hat{E}_X^{(1, 0)} = \frac{1}{2 \gamma} \frac{\bar{p}_X}{\bar{p}_Z}  \hat{Q}_Z^{(1, 0)}. 
\ee
By applying Azuma's inequality (see Appendix \ref{Azuma}), 
\be
E_X^{(0, 0)} \sim \frac{1}{2 \gamma}  Q_Z^{(0, 0)},~~E_X^{(0, 1)} \sim \frac{1}{2 \gamma} \frac{\bar{p}_X}{\bar{p}_Z} Q_Z^{(0, 1)},~~E_X^{(1, 0)} 
\sim \frac{1}{2 \gamma} \frac{\bar{p}_X}{\bar{p}_Z}  Q_Z^{(1, 0)} 
\label{EXandQX}
\ee
holds in the asymptotic limit. 
By combining Eqs.~(\ref{Xerrorineq}) and (\ref{EXandQX}), we obtain 
\begin{align}
\bs
\frac{{\bar{p}_Z}^2}{{\bar{p}_X}^2} E_X^{(1, 1)} 
&\lesssim \frac{{\bar{p}_Z}^2}{{\bar{p}_X}^2} \left(E_X - \frac{1}{2 \gamma} Q_Z^{(0, 0)} - \frac{1}{2 \gamma} \frac{\bar{p}_X}{\bar{p}_Z} Q_Z^{(0, 1)} - \frac{1}{2 \gamma} \frac{\bar{p}_X}{\bar{p}_Z} Q_Z^{(1, 0)}\right) \\
& = \frac{{\bar{p}_Z}^2}{{\bar{p}_X}^2} E_X - \frac{1}{2 \gamma} \frac{{\bar{p}_Z}^2}{{\bar{p}_X}^2} Q_Z^{(0,0)} - \frac{1}{2 \gamma} \frac{\bar{p}_Z}{\bar{p}_X} Q_Z^{(0, 1)} - \frac{1}{2 \gamma} \frac{\bar{p}_Z}{\bar{p}_X} Q_Z^{(1, 0)}.
\es
\label{preconcav}
\end{align}
From these considerations, we obtain  
\begin{align}
\bs
&\left( Q_Z^{(0,0)} + Q_Z^{(0,1)} + Q_Z^{(1,0)} +Q_Z^{(1,1)} h\left( \frac{\bar{p}_Z^2 \gamma}{\bar{p}_X^2} \frac{E_X^{(1,1)}}{Q_Z^{(1,1)}}\right) \right) /q_Z \\
&= \frac{Q_Z^{(0,0)} + Q_Z^{(0,1)} + Q_Z^{(1,0)}} {q_Z} h\left( \frac{1}{2} \right) + \frac{Q_Z^{(1,1)}}{q_Z} h\left( \frac{\bar{p}_Z^2 \gamma}{\bar{p}_X^2} \frac{E_X^{(1,1)}}{Q_Z^{(1,1)}}\right) \\
& \leq h\left(\frac{Q_Z^{(0,0)} + Q_Z^{(0,1)} + Q_Z^{(1,0)}} {2q_Z} + \frac{\bar{p}_Z^2 \gamma}{\bar{p}_X^2} \frac{E_X^{(1,1)}}{q_Z} \right) \\
& \lesssim h\Bigg( \frac{{\bar{p}_Z}^2 \gamma}{{\bar{p}_X}^2} \frac{E_X}{q_Z} + \frac{1}{2q_Z} \left(1 - \frac{{\bar{p}_Z}^2 }{{\bar{p}_X}^2} \right)Q_Z^{(0,0)} + \\
&~~~~~~~~~~~~~~~~~~~~~~~~~~~~\frac{1}{2q_Z} \left(1 - \frac{\bar{p}_Z}{\bar{p}_X} \right) (Q_Z^{(0,1)} + Q_Z^{(1,0)})  \Bigg) \\
& \leq h\Bigg( \frac{{\bar{p}_Z}^2 \gamma}{{\bar{p}_X}^2} \frac{E_X}{q_Z}\Bigg),
\es
\end{align}
where the first inequality follows from the concavity of the function $h(x)$, while the second inequality follows from Eq.~(\ref{preconcav}). 
The second and third inequalities use the fact that $h(x)$ is a non-decreasing continuous function in the range of $0\leq x \leq 1/2$. 
From the assumption of lemma \ref{lemmaconcavity}, the terms in the brackets are no greater than 1/2.  
The third inequality uses $\bar{p}_Z / \bar{p}_X \geq 1$ from the assumption of the lemma. 
This result leads to Eq.~(\ref{eqlemma2}).

\end{proof}

By combining Lemma~\ref{lemmasingle}, \ref{lemmamulti} and \ref{lemmaconcavity}, we obtain the main theorem. 

\begin{thm} 
If the relations
\be
\bar{p}_Z \geq \bar{p}_X, ~~
\frac{{\bar{p}_Z}^2 \gamma}{{\bar{p}_X}^2} \frac{E_X}{Q_Z - \alpha( Q_{\bot_A, Z_B} + Q_{Z_A, \bot_B})} \leq \frac{1}{2} 
\label{pbarineq} 
\ee 
hold,  
the lower bound on the secure key rate $R$ per round of the passive BBM92 protocol in the asymptotic limit is given by 
\begin{align}
\begin{split}
&R \gtrsim
(Q_Z - \alpha (Q_{\bot_A, Z_B} + Q_{Z_A, \bot_B})) \\ 
&~~~~~~~~~~~~~~~~~~~~~~~~~\left(1-h\left(\frac{(\bar{p}_Z^2 \gamma/\bar{p}_X^2)E_X}{Q_Z - \alpha (Q_{\bot_A, Z_B} + Q_{Z_A, \bot_B})}\right)  \right) - Q_Z f_{\rm EC}. 
\end{split}
\label{finalkeyrate}
\end{align}
 \label{thm:1}
\end{thm}

\begin{proof}
Applying Lemma \ref{lemmasingle} to Eq.~(\ref{PAamount}), we obtain 
\begin{align}
\bs
Q_Z f_{\rm PA} \sim \sum_{\substack{n_A, n_B \\ n_A \ge 2 \lor n_B \ge 2}} Q_Z^{(n_A, n_B)} + Q_Z^{(0,0)} + Q_Z^{(0,1)} + Q_Z^{(1,0)} \\
~~~~~~~~~~~~~~~~~~~~~~~~~~~~~~~~~~~~~+Q_Z^{(1,1)} h\left(\frac{{\bar{p}_Z}^2 \gamma}{{\bar{p}_X}^2} \frac{E_X^{(1,1)}}{Q_Z^{(1,1)}}\right) 
\label{PAbound}
\es
\end{align}
Applying Lemma \ref{lemmaconcavity} to Eq.~(\ref{PAbound}), we further obtain 
\be
\begin{split}
Q_Z f_{\rm PA} &\lesssim \sum_{\substack{n_A, n_B \\ n_A \ge 2 \lor n_B \ge 2}} Q_Z^{(n_A, n_B)} + q_Z \left(\frac{\bar{p}_Z^2 \gamma}{\bar{p}_X^2} \frac{E_X}{q_Z} \right) \\ 
&= Q_Z - q_Z + q_Z h \left(\frac{\bar{p}_Z^2 \gamma}{\bar{p}_X^2} \frac{E_X}{q_Z} \right),
\end{split}
\ee
if $\bar{p}_Z \geq \bar{p}_X$ and 
\be
\frac{\bar{p}_Z^2 \gamma}{\bar{p}_X^2} \frac{E_X}{q_Z} \leq 1/2 \label{assulemma2}
\ee
are satisfied. 
The secure key rate per round is then given by 
\be
R = Q_Z(1 - f_{\rm PA} - f_{\rm EC}) \gtrsim q_Z \left(1- h \left( \frac{\bar{p}_Z^2 \gamma}{\bar{p}_X^2} \frac{E_X}{q_Z} \right) \right) - Q_Z f_{\rm EC}. 
\label{Rqz}
\ee
From Eq.~(\ref{qz}) and Lemma \ref{lemmamulti}, we have 
\be
q_Z = Q_Z - \sum_{\substack{n_A, n_B \\ n_A \ge 2 \lor n_B \ge 2}} Q_Z^{(n_A, n_B)} \gtrsim Q_Z - \alpha( Q_{\bot_A, Z_B} + Q_{Z_A, \bot_B}). \label{qZlower}
\ee
From this bound, Eq.~(\ref{assulemma2}) is satisfied if 
\be
\frac{{\bar{p}_Z}^2 \gamma}{{\bar{p}_X}^2} \frac{E_X}{Q_Z - \alpha( Q_{\bot_A, Z_B} + Q_{Z_A, \bot_B})} \leq \frac{1}{2} \label{keycondition}
\ee 
holds. 
Applying the bound Eq.~(\ref{qZlower}) to Eq.~(\ref{Rqz}), the secure key rate is lower bounded by  
\begin{align}
\begin{split}
&R \gtrsim
(Q_Z - \alpha (Q_{\bot_A, Z_B} + Q_{Z_A, \bot_B})) \\ 
&~~~~~~~~~~~~~~~~~~~~~~~~~\left(1-h\left(\frac{(\bar{p}_Z^2 \gamma/\bar{p}_X^2)E_X}{Q_Z - \alpha (Q_{\bot_A, Z_B} + Q_{Z_A, \bot_B})}\right)  \right) - Q_Z f_{\rm EC}
\end{split}
\end{align}
if Eq.~(\ref{keycondition}) is satisfied. 
\end{proof}

\subsection{Numerical analysis}
Here, we present the results of numerical simulation of the key rate given in Theorem~\ref{thm:1}. 
We show the physical models and derive analytical expressions of the observed quantities, and then present the numerical analysis. 

\begin{figure}[t]
 \centering
 \includegraphics[keepaspectratio, scale=0.4]
      {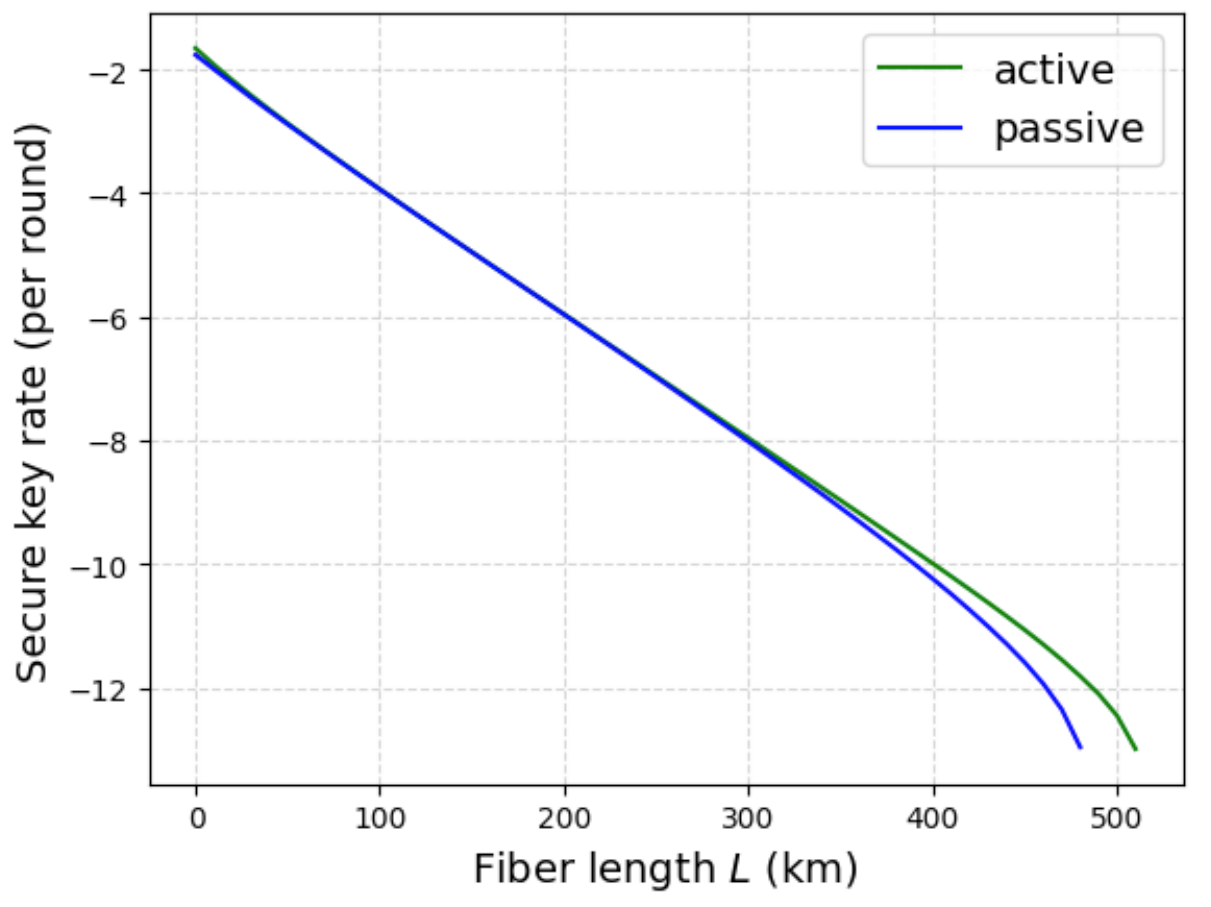}
 \caption{Secure key rate per round as a function of the fiber length $L$ between Alice and Bob. The channel transmittance corresponding to source-Alice (or source-Bob) is modeled as $\eta_{\rm ch}=10^{-l/50}$, where the source is placed midway between Alice and Bob ($L=2l$), and the overall transmittance is $\eta=\eta_{\rm ch}\eta_{\rm det}^Z$. The detector efficiencies are $\eta_{\rm det}^Z=\eta_{\rm det}^X=0.7$, the dark count probability is $d=10^{-7}$, the detector imperfection and misalignment error rate is $e_d=0.03$, and the error-correction efficiency is 1.16. We optimize the signal intensity $\mu$ and the basis choice probability $p_Z$ at each distance. The blue curve shows the secure key rate $R$ of the proposed protocol with $r=\eta_X/\eta_Z=1$, while the green curve shows that of the active-BBM92 protocol for comparison. The vertical axis is shown on a base-10 logarithmic scale.}
 \label{Fig:keyrate}
\end{figure}

\subsubsection{Physical models and expressions for observed quantities}
One of the most widely adopted models for entanglement-based QKD is the model described in Ref.~\cite{2007Ma}, which assumes a parametric down-conversion (PDC) source pumped by a pulsed laser.
Although we also adopt this model, we modify the expressions for the observed parameters to account for the passive measurement setup rather than the active one.
This enables a direct comparison of the key rates of the passive and active protocols as a function of the communication distance.
The details are provided in Appendices~\ref{appe: QW} and \ref{appe:EW}.

According to \cite{2007Ma}, the state emitted from a type-II PDC source is written as 
\begin{align}
\ket{\Psi}_{AB} &= \sum_{n=0}^{\infty} \sqrt{P(n)} \ket{\Phi_n}_{AB},\\
 \ket{\Phi_n}_{AB} &= \frac{1}{\sqrt{n+1}} \sum_{m=0}^{n} (-1)^m \ket{n-m,m}_A \ket{m,n-m}_B. \label{Phi_n}
\end{align}
Here, $P(n)$ is the probability that $n$-photon pair is generated, which is characterized with a parameter $\lambda$ as 
\be
P(n) = \frac{(n+1)\lambda^{n}}{(1+\lambda)^{n+2}},
\ee
where $\mu = 2 \lambda$ is the average photon number of photon pairs generated by one pump pulse. 
The normalized state $\ket{n-m, m}_A$ denotes that $n-m$ and $m$ photons have horizontal and vertical polarizations, respectively. 
As shown in Appendix~\ref{Appe:basisrotation}, the state $\ket{\Phi_n}_{AB}$ retains the same form when expressed in the diagonal polarization basis~\cite{2012PanRMP}:
\be
\label{Phiequiv}
\ket{\Phi_n}_{AB} =
\frac{(-1)^n}{\sqrt{n+1}}
\sum_{m=0}^{n}
(-1)^m
\ket{n-m,m}_{A_D}
\ket{m,n-m}_{B_D},
\ee
where the state $\ket{n-m,m}_{A_D}$ denotes that $n-m$ and $m$ photons have diagonal and anti-diagonal polarizations, respectively.
Since the entanglement state $\ket{\Psi}_{AB}$ is anti-correlated, the bit flip operation over all bits is necessary by Alice or Bob. 

Another physical assumption we adopt is the `source-in-the-middle' scenario, that is, the channel transmittances $\eta_A$ and $\eta_B$ corresponding to the source-Alice and source-Bob, respectively, equal $\eta$. 

Using the models described above, in Appendix \ref{appe: QW} and \ref{appe:EW}, we derive closed-form expressions for the expected values $\bar{Q}_Z$, $\bar{Q}_X$, $\bar{Q}_{Z_A, \bot_B}$, $\bar{Q}_{\bot_A, Z_B}$, $\bar{E}_X$ and $\bar{e}_Z$ of the observed quantities $Q_Z$, $Q_X$, $Q_{Z_A, \bot_B}$, $Q_{\bot_A, Z_B}$, $E_X$ and $e_Z$, respectively, as a function of $\eta$. 
These closed-form expressions uniquely determine the key rate without assuming the cut-off for the photon number in the simulation.

\subsubsection{Simulation results}
We present the results of numerical calculation of the key rate $R$ per round given by Eq. (\ref{finalkeyrate}). 
In the simulation shown in Fig.~\ref{Fig:keyrate}, we assume detectors' quantum efficiency $\eta_{\rm det}^Z = \eta_{\rm det}^X = 0.7$ and dark count probability $d = 10^{-7}$, which are achievable with commercial SSPDs \cite{2024Sanari}. 
The channel transmittance of the optical fiber $\eta_{\rm ch}$ connecting the source and Alice (or Bob) is modeled as $\eta_{\rm ch} = 10^{-l/50}$ where $l$ indicates its length (km). The overall transmittance is $\eta = \eta_{\rm ch} \eta_{\rm det}^Z$. 
The error rate due to the detector imperfections and misalignment is $e_d = 0.03$.  
The cost for error correction $f_{\rm EC}(e_Z)$ is $1.16\times h(e_Z)$~\cite{2024Luo}. 
Regarding the non-sampling probability $\gamma$, we assume $\gamma \to 1$ because sampling ratios can be negligibly small in the asymptotic limit of $m_{\rm rep} \to \infty$. 
In Fig.~\ref{Fig:keyrate}, we plot numerical results  for two cases, each optimized over the signal intensity $\mu$ and the basis choice probability $p_Z$ at each distance: (i) secure key rate of the active-BBM92 protocol for comparison (green), (ii) our result $R$ with $r = 1$ (blue). 
The horizontal axis is fiber-length $L$ between Alice and Bob, which equals $2l$ in the `source-in-the-middle' scenario. 
The vertical axis is on a logarithmic scale with base 10.


As shown in Fig.~\ref{Fig:keyrate}, the key rates of the active case and passive case are almost identical except the small gap in the range of long communication distances. 
This gap stems not from looseness our bound but from the sensitivity of the passive protocol to dark counts, which is also discussed for the passive BB84 protocol in \cite{2025Kawakami}. 
Therefore, our proof does not sacrifice the key rate, and consequently, the performance of the passive BBM92 protocol is shown to be comparable to the active protocol, despite the fact that the passive protocol is composed of simple setups.

\section{Security of passive QCKA protocol}
\label{Sec:QCKA}
In this section, we establish the security of the passive QCKA for an arbitrary number of parties. First, we describe the protocol and define the observed quantities used in the security analysis. We then extend the framework developed for the passive BBM92 protocol to the multipartite setting and derive a secure key rate for the passive QCKA.

\subsection{Passive QCKA Protocol}
\label{Sec: QCKAprotocol}
Here, we introduce a polarization-type QCKA protocol, in which $N+1$ legitimate parties select their basis passively, while an untrusted party, Fred, is presumed to generate and distribute GHZ states.
The legitimate parties consist of Alice and $N$ Bobs, where Alice is labeled by $j=0$ and the $j'$-th Bob is labeled by $j=j'$ for $j'=1,\ldots,N$.
Note that when $N=1$, the protocol is identical to the BBM92 protocol in Sec.~\ref{Sec: BBMprotocol}.
Similarly to the BBM92 protocol, each party uses two complementary measurement bases ($Z$ and $X$), where the secret key is extracted from the $Z$ basis and the signal disturbance is monitored using the $X$ basis.
The protocol is described as follows.
\\

\noindent
(1) {\it State Preparation}: 
Fred is presumed to prepare a $(N+1)$-partite (not necessarily perfect) GHZ state and distribute subsystems among legitimate parties. 
\\
(2) {\it Measurement}: 
Alice and Bobs perform the same procedure with the same setup as in Step (2) in Sec.~\ref{Sec: BBMprotocol}, and they obtain a basis label $W_j\in \{Z,X,\bot,\emptyset\}$ and a bit $b_j \in \{0,1\}$ for $j=0,\ldots,N$.
\\
(3) {\it Repetition}: 
Alice, Bobs, and Fred repeat Steps (1) and (2) $m_{\rm rep}$ times. 
\\
(4) {\it Public communication}: 
For each of the $m_{\rm rep}$ rounds, Alice and Bobs publicly announce $\{W_j\}$.
They also disclose $\{b_j\}$ for rounds where $W_j=X$ for all $j=0,\ldots,N$.
For rounds where $W_j=Z$ for all $j=0,\ldots,N$, Alice and Bobs randomly select $s=1$ (sampling) with probability $1-\gamma$ and $s=0$ (non-sampling) with probability $\gamma$.
For rounds with $s=1$, they disclose the corresponding bits $\{b_j\}$.
\\
(5) {\it Parameter estimation}:
Alice and Bobs obtain bit strings of lengths
\begin{equation}
m_Z \coloneqq m\left( \{W_i = Z\}_{i=0}^{N}, s=0\right), 
\end{equation}
which corresponds to the sifted key. 
For the $Z$-basis sampling rounds, they obtain the number of sampled rounds, 
\be
m_Z^{\rm sam} \coloneqq m\left(\{W_i=Z\}_{i=0}^N, s=1\right),
\ee
and the number of bit errors,
\begin{equation}
m_{Z_j}^{\rm sam,err}
\coloneqq
m\left(
\{W_i = Z\}_{i=0}^{N},
\, s=1,
\, b_0 \neq b_j
\right),
\end{equation}
between Alice and the $j$-th Bob. 
For the $X$-basis rounds, they obtain the number of total rounds, 
\be
m_X \coloneqq m\left(\{W_i=X\}_{i=0}^N\right),
\ee
and the number of error rounds,
\begin{equation}
m_{X}^{\rm err}
\coloneqq
m\left(
\{W_i = X\}_{i=0}^{N},
\bigoplus_{i=0}^N b_i = 1
\right), 
\end{equation}
where `$\bigoplus$' denotes addition modulo 2. 
Alice and Bobs also obtain the numbers of rounds
\begin{equation}
m_{\bot_j,Z}
\coloneqq
m\left(
W_j = \bot,
\{W_i = Z\}_{i\neq j}
\right)
\end{equation}
for all $j=0,\ldots,N$, corresponding to the events where a cross click occurs on the $j$-th party and the $Z$ basis is chosen on all remaining parties.
\\
(6) {\it Error correction}: 
Alice generates a syndrome of length $m_Z f_{\rm EC}$ from the sifted key. 
She encrypts this syndrome using a pre-shared secret key with the Bobs and sends it to them~\cite{2006Koashi}. 
The Bobs then correct their keys according to the decrypted syndrome.
\\
(7) {\it Privacy amplification}: 
Alice and Bobs perform privacy amplification by shortening their sifted keys by $m_Z f_{\rm PA}$ bits to obtain the final secret key.
\\

Let us define the following ratios based on the numbers obtained in the above protocol: 
\begin{align}
Q_Z \coloneqq \frac{m_Z}{m_{\rm rep}},~
E_X \coloneqq\frac{m_X^{\rm err}}{m_{\rm rep}}, 
~Q_{\bot_j, Z}\coloneqq \frac{m_{\bot_j, Z}}{m_{\rm rep}}, 
\label{observed}
\end{align}
and 
\be
e_{Z_j} \coloneqq \frac{m_{Z_j}^{\rm sam,err}}{m_Z^{\rm sam}},
\ee
where $j=0,\ldots,N$ for quantities indexed by $j$. By its definition, $e_{Z_0}=0$. 
In this paper, we consider the asymptotic secure key rate in the limit $m_{\rm rep}\to \infty$. 
In this limit, the asymptotic value of $f_{\rm PA}$ is expressed as a function of $Q_Z, E_X$ and $\{Q_{\bot_j, Z}\}$. 
For one-way classical communication, the cost for error correction $f_{\rm EC}$ is given by~\cite{2018Grasselli}
\be
f_{\rm EC}(\{e_{Z_j}\}) = c_{\rm EC}\cdot \max_j h(e_{Z_j}),
\ee
where $c_{\rm EC}~(\geq 1)$ denotes the error correction efficiency factor. 
The key rate $R$ per round is given by
\begin{equation}
R= Q_{Z}\left(1-f_{\rm PA}( Q_Z, E_X , \{Q_{\bot_j, Z}\}) -f_{\rm EC}(\{e_{Z_j}\}) \right).
\end{equation}

\subsection{Security Proof}
\label{Sec:QCKAsecurity}
Here, we generalize the security proof of the passive BBM92 protocol to the passive QCKA protocol with an arbitrary number of parties. We show that the virtual-qubit description and the phase-error analysis developed for BBM92 can be extended to the multipartite GHZ-state setting. This leads to a security proof of the passive QCKA, whose main result is summarized in Theorem~\ref{Theo:QCKA}.

Similarly to the security proof of the BBM92 protocol, our starting point is to introduce an alternative protocol which is equivalent to the actual protocol from the viewpoint of Eve. 
This alternative protocol is obtained by replacing Step (2) of the actual protocol with Step (2') in Sec.~\ref{Sec:alternativeBBM}, with the only modification that the single Bob is replaced by $N$ Bobs. 
In the alternative protocol, photon numbers 
\be
\vec{n}=(n_0, \ldots ,n_N)
\ee
for each party are treated as observed quantities, which allows us to define the following ratios: 
\begin{equation}
\begin{split}
Q_{Z}^{(\vec{n'})} &\coloneqq m\left(\{W_i=Z\}_{i=0}^N, s=0, \vec{n}=\vec{n'}\right) / m_{\rm rep}, \\
E_{X}^{(\vec{n'})} &\coloneqq m\left(\{W_i=X\}_{i=0}^N, \vec{n}=\vec{n'}, \bigoplus_{i=0}^N b_i = 1
\right) / m_{\rm rep}, \\
Q_{\bot_j, Z}^{(\vec{n'})} &\coloneqq m\left(W_j=\bot, \{W_i = Z\}_{i\neq j}, \vec{n}=\vec{n'} \right) / m_{\rm rep} .
\label{alternativeQEQCKA}
\end{split}
\end{equation}

The virtual qubit on system $A$ and the phase error are also defined in a similar manner to the BBM92 protocol as follows:
as a consequence of the relabeling of Alice's system, let
\be
(\ket{H}_0,\ket{V}_0,\ket{D}_0,\ket{\bar{D}}_0) = (\ket{H}_A,\ket{V}_A,\ket{D}_A,\ket{\bar{D}}_A)
\ee
and 
\be
\hat{K}_{0,l} = \hat{K}_{A,l}~~(l=0,1,2,3).
\ee
For Alice's photon number $n_0 \neq 1$, if she obtains an outcome for the case of $W_0=Z$ and $b_0=0~(1)$, she prepares a single photon in the $Z$-basis state $\ket{H}_0~(\ket{V}_0)$ as the virtual qubit on system $A$. 
For $n_0=1$, Alice first applies the filter with success probability $(1-d)^2$ and the CPTP map specified by $\{\hat{K}_{0,l}\}$ to her received single photon, and keeps the output as the virtual qubit. 
For the events with ($\{W_i=Z\}_{i=0}^N, s=0$), where all parties select $Z$ basis and the choice of non-sampling is made, the $Z$-basis measurement $\{\ket{H}\bra{H}_0,  \ket{V}\bra{V}_0\}$ on Alice's virtual qubit is equivalent to the actual process of obtaining a sifted key. The phase error is then defined in the following virtual scenario: suppose that the $X$-basis measurement $\{\ket{D}\bra{D}_0,  \ket{\bar{D}} \bra{\bar{D}}_0 \}$ is performed on 
the virtual qubit for the events where $\{W_i=Z\}_{i=0}^{N}$ to obtain the outcome $\tilde{b}_0$ instead of the $Z$-basis measurement, and the Bobs cooperate and try to guess the outcome by using all of their quantum states. 
By defining $\tilde{b}_B$ as a guess bit by the Bobs, the phase error is then defined as a bit error between $\tilde{b}_0$ and $\tilde{b}_B$.

For the events with $n_0 \neq 1$, the phase error rate is 1/2 because $X$-basis measurement is performed on the prepared $Z$-basis state $\ket{H}_0$ or $\ket{V}_0$ and the Bobs cannot guess the outcome at all.  
In addition, we assume the worst case where the Bobs give up their guess for the outcome when there exists $j\in\{1,\ldots,N\}$ which satisfies $n_j \neq 1$, corresponding to a phase error rate of 1/2 in those cases. 
These imply that in our security proof, the secret key is extracted solely from the events with $\vec{n}=(1,1,\ldots,1)$, while all other events are discarded through privacy amplification. 
For such events with $\vec{n}=(1,1,\ldots,1)$, we define 
\begin{align}
E_{\rm ph}^{(1,..,1)}& \coloneqq m\left(\{W_i=Z\}_{i=0}^{N}, s=0, \vec{n}=(1,1,\ldots,1), \tilde{b}_0 \neq \tilde{b}_B \right) / m_{\rm rep}.
\label{phaseerrorratio}
\end{align}
Then the phase error rate is represented as  $E_{\rm ph}^{((1,..,1)} / Q_Z^{(1,..,1)}$. 
Since all events except those with $\vec{n}= (1,1,\ldots,1)$ have the phase error rate 1/2, the fraction of bits to be sacrificed in privacy amplification is given by 
\begin{align}
\begin{split}
Q_Z f_{\rm PA} &= 
\sum_{\substack{ \vec{n} \\ \vec{n}\neq (1,..,1) }} Q_Z^{(\vec{n})}
+Q_Z^{(1,..,1)} h\left( \frac{E_{\rm ph}^{(1,..,1)}}{Q_Z^{(1,..,1)}}\right) \\
& = \sum_{\substack{ \vec{n} \\ \bigvee_{i=0}^{N}(n_i\geq2) }} Q_Z^{(\vec{n})} 
+ \sum_{\substack{\vec{n} \\ \vec{n} \in \{0,1\}^{N+1},~\vec{n}\neq (1,..,1) }} Q_Z^{(\vec{n})}\\
&~~~~~~~~~~~~~~~~~~~~~~~~~~~~~~~~~~+Q_Z^{(1,..,1)} h\left( \frac{E_{\rm ph}^{(1,..,1)}}{Q_Z^{(1,..,1)}}\right). \label{PAamountQCKA}
\end{split}
\end{align}

Similarly to Sec.~\ref{lemmaproof}, we derive an upper bound on $f_{\rm PA}$ and obtain a lower bound on the secure key rate.
The result is summarized in Theorem~\ref{Theo:QCKA}.
The proof follows the same structure as that of Theorem~\ref{thm:1}. Specifically, it is based on three lemmas, which generalize Lemmas~\ref{lemmasingle}, \ref{lemmamulti}, and \ref{lemmaconcavity}, respectively. The theorem is then established by combining these lemmas.

We begin with the following lemma, which provides an upper bound on the phase error ratio for single-photon events and generalizes Lemma~\ref{lemmasingle}.

\begin{lemma}
\label{lemmasingleQCKA}
The relation
\be
E^{(1,..,1)}_{\rm ph} \sim \frac{{\bar{p}_Z}^{N+1} \gamma}{{\bar{p}_X}^{N+1}} E^{(1,..,1)}_X
\ee
holds.
\end{lemma}

\begin{proof}
As already mentioned, the phase error is defined as the discrepancy between Alice's virtual $X$-basis measurement outcome $\tilde{b}_0$ and the Bobs' guess $\tilde{b}_B$. 
Here, we describe how the Bobs determine $\tilde{b}_B$ in order to evaluate the phase error rate.
On the system of $j$-th Bob, which is the input system to the beam splitter in Fig.~\ref{setup2}, we define  $\ket{H}_j$, $\ket{V}_j$, $\ket{D}_j$ and $\ket{\bar{D}}_j$ as single-photon states with horizontal, vertical, diagonal and anti-diagonal polarizations, respectively. 
Then, the same discussion in the proof of Lemma~\ref{lemmasingle} holds for each Bob. Namely, for events with $n_j=1$, the actual measurement by $j$-th Bob is equivalent to the process where the filtering operation with success probability $(1-d)^2$ and CPTP map specified by $\{\hat{K}_{j,l}\}_{l=0,1,2,3}$ in Eq.~(\ref{Krauses}) are applied to the received single photon, followed by the measurement $\{\ket{H}\bra{H}_j, \ket{V}\bra{V}_j\}$ or $\{\ket{D}\bra{D}_j, \ket{\bar{D}}\bra{\bar{D}}_j\}$ according to his active basis choice $Z$ or $X$, respectively. 

Here, we define the Bobs' guess $\tilde{b}_B$ for events with $\{n_i=1\}_{i=1}^N$ as
\begin{equation}
\tilde{b}_B
=
\bigoplus_{i=1}^N \tilde{b}_i,
\end{equation}
where $\tilde{b}_i$ is the outcome of the virtual $X$-basis measurement performed by the $i$-th Bob on the single photon.
Accordingly, the phase error for events with $\vec{n}=(1,1,\ldots,1)$ is described as follows.
Alice and Bobs first independently apply a filter with success probability $(1-d)^2$, followed by the CPTP maps $\{\hat{K}_{i,l}\}_{i=0,\ldots,N,\ l=0,1,2,3}$.
They then independently choose the $Z$ basis with probability $\bar{p}_Z$ and select the events as non-sampling with probability $\gamma$.
For the resulting events, the phase error is defined as the discrepancy between Alice's $X$-basis measurement outcome $\tilde{b}_0$ and the modulo-2 sum $\tilde{b}_B$ of the Bobs' $X$-basis measurement outcomes $\{\tilde{b}_i\}_{i=1}^{N}$ on the single photons.
Equivalently, the phase error event is given by
\begin{equation}
\tilde{b}_0 \oplus \tilde{b}_B = \bigoplus_{i=0}^N \tilde{b}_i = 1.
\end{equation}

Then, the POVM element $\hat{E}_{\rm ph}^{(1,..,1)} $ corresponding to obtain the phase error event with ($\{W_i = Z\}_{i=0}^{N}, s= 0, \vec{n}=(1,1,\ldots,1), \tilde{b}_0 \neq \tilde{b}_B$) is represented as 
\be
\hat{E}_{\rm ph}^{(1,..,1)}  =   \sum_{\substack{
\tilde{b}_0,\ldots,\tilde{b}_N\in\{0,1\}\\
\bigoplus_{i=0}^N \tilde b_i = 1
}} \gamma \left( \bigotimes_{i=0}^N   \hat{G}_{{\rm ph}_i \tilde{b}_i}^{(1)} \right),
\label{POVMEphQCKA}
\ee
where 
\be
\begin{split}
&\hat{G}_{{\rm ph}_i0}^{(1)} \coloneqq  \bar{p}_Z (1-d)^2 \sum\limits_{l=0}^3  \hat{K}_{i,l}^{\dagger} \ket{D}\bra{D}_i \hat{K}_{i,l}, \\
&\hat{G}_{{\rm ph}_i1}^{(1)} \coloneqq  \bar{p}_Z (1-d)^2 \sum\limits_{l=0}^3  \hat{K}_{i,l}^{\dagger} \ket{\bar{D}}\bra{\bar{D}}_i \hat{K}_{i,l}. \\
\end{split}
\ee
On the other hand, the procedure for obtaining the $X$-basis error event with ($\{W_i = X\}_{i=0}^{N}, \vec{n}=(1,..,1), \bigoplus_{i=0}^{N} b_i =1$) is equivalent to that for the phase error, except that the each party chooses the $X$ basis instead of the $Z$ basis and no sampling procedure is required. 
Specifically, the POVM element $\hat{E}_X^{(1,..,1)} $ corresponding to the $X$-basis error event is represented as  
\be
\hat{E}_{X}^{(1,..,1)}  =   \sum_{\substack{
b_0,\ldots,b_N\in\{0,1\}\\
\bigoplus_{i=0}^N b_i = 1
}}  \left( \bigotimes_{i=0}^N   \hat{G}_{X_i b_i}^{(1)} \right),
\label{POVMEXQCKA}
\ee
where we defined 
\be
\begin{split}
&\hat{G}_{X_i0}^{(1)} \coloneqq  \bar{p}_X (1-d)^2 \sum\limits_{l=0}^3  \hat{K}_{i,l}^{\dagger} \ket{D}\bra{D}_i \hat{K}_{i,l}, \\
&\hat{G}_{X_i1}^{(1)} \coloneqq  \bar{p}_X (1-d)^2 \sum\limits_{l=0}^3  \hat{K}_{i,l}^{\dagger} \ket{\bar{D}}\bra{\bar{D}}_i \hat{K}_{i,l}.
\end{split}
\ee
Eqs.~(\ref{POVMEphQCKA}) and (\ref{POVMEXQCKA}) lead to the ralation
\be
\hat{E}^{(1,..,1)}_{\rm ph} = \frac{{\bar{p}_Z}^{N+1} \gamma}{{\bar{p}_X}^{N+1}} \hat{E}^{(1,..,1)}_X.
\ee
By applying Azuma's inequality (see Appendix \ref{Azuma}), we obtain
\be
E^{(1,..,1)}_{\rm ph} \sim \frac{{\bar{p}_Z}^{N+1} \gamma}{{\bar{p}_X}^{N+1}} E^{(1,..,1)}_X
\ee
in the asymptotic limit. 
\end{proof}

Next, the following lemma gives an upper bound on the contribution from multi-photon events and generalizes Lemma \ref{lemmamulti}. 

\begin{lemma}
\label{lemmamultiQCKA}
The relation
\be
\sum_{\substack{ \vec{n}\\ \bigvee_{i=0}^{N}(n_i\geq2) }} Q_Z^{(\vec{n})} \lesssim \alpha \sum\limits_{i=0}^{N} Q_{\bot_i, Z}
\label{templemmamultiQCKA}
\ee
holds. 
\end{lemma}

\begin{proof}

The left-hand side in Eq.~(\ref{templemmamultiQCKA}) can be upper bounded as 
\be
\sum_{\substack{ \vec{n}\\ \bigvee_{i=0}^{N}(n_i\geq2) }} Q_Z^{(\vec{n})}
 \leq \sum\limits_{i=0}^{N}  Q_Z^{(n_i\geq 2)} ,
\label{nAnBupperQCKA}
\ee
where we defined 
\be
Q_Z^{(n_i\geq 2)} \coloneqq \sum_{\substack{\vec{n}\\ n_i\geq 2}} Q_Z^{(\vec{n})}.
\ee
Similarly to Eq.~(\ref{FdashA}), for each $j \in\{0,\ldots,N\}$ and $n_j \geq 1$, the POVM elements corresponding to the outcomes $W_j = Z$, $W_j = X$ and $W_j = \bot$ are given by 
\begin{align}
\bs
&\hat{F'}_{Z_j}^{(n_j)} = (1-d)^2 {\bar{p}_Z}^{n_j} \hat{\mathbbm{1}}_j^{(n_j)},~~~\hat{F'}_{X_j}^{(n_j)} = (1-d)^2 {\bar{p}_X}^{n_j} \hat{\mathbbm{1}}_j^{(n_j)}, \\
&\hat{F'}_{\bot_j}^{(n_j)} = \left( 1 - (1-d)^2 (\bar{p}_Z^{n_j} + \bar{p}_X^{n_j}) \right) \hat{\mathbbm{1}}_j^{(n_j)},
\label{FdashAQCKA}
\es
\end{align}
respectively, where $\hat{\mathbbm{1}}_j^{(n_j)}$ denotes the projection operator on the $n_j$-photon subspace associated with the $j$-th party. 
For $n_j = 0$, the POVM elements are given by
\be
\begin{split}
\hat{F'}_{Z_j}^{(0)} &= \hat{F'}_{X_j}^{(0)} = (1-d)^2 (1-(1-d)^2) \ket{\rm vac}\bra{\rm vac}_j,\\
\hat{F'}_{\bot_j}^{(0)} &=  (1 - (1-d)^2)^2  \ket{\rm vac}\bra{\rm vac}_j, 
\label{FdashA0QCKA}
\end{split}
\ee
where $\ket{\rm vac}_j$ is the vacuum state on system $j$. 

Here, for arbitrary $\vec n$ and each $j\in\{0,\ldots,N\}$, we define the following operators corresponding to the sifted-key and cross-click events:
\begin{equation}
\begin{split}
\hat{Q}_Z^{(\vec n)}
&\coloneqq
\gamma \bigotimes_{i=0}^{N} \hat{F'}_{Z_i}^{(n_i)},
\\
\hat{Q}_{\bot_j,Z}^{(\vec n)}
&\coloneqq
\hat{F'}_{\bot_j}^{(n_j)} \bigotimes_{i\neq j} \hat{F'}_{Z_i}^{(n_i)},
\end{split}
\label{hatQnQCKA}
\end{equation}
where $\gamma$ denotes the non-sampling probability.
By using the expressions in Eq.~(\ref{FdashAQCKA}), we obtain 
\be
\begin{split}
\hat{Q}_Z^{(\vec{n})} & =  \frac{ \bar{p}_Z^{n_j} (1-d)^2 \gamma  }
{1 - (\bar{p}_Z^{n_j} + \bar{p}_X^{n_j})(1-d)^2}  \hat{Q}_{\bot_j, Z}^{(\vec{n})} \\
& \leq \frac{ \bar{p}_Z^{n_j} \gamma  }
{1 - (\bar{p}_Z^{n_j} + \bar{p}_X^{n_j})}  \hat{Q}_{\bot_j, Z}^{(\vec{n})} 
\label{operatorrelationQCKA}
\end{split}
\ee
for $n_j \geq 1$. 
For $n_j \geq 2$, 
the relation
\be
\hat{Q}_Z^{(\vec{n})} \leq \alpha \hat{Q}_{\bot_j, Z}^{(\vec{n})} 
\ee
is obtained, where $\alpha$ is defined in Eq.~(\ref{alpha}). 
Therefore, 
\be
\sum_{\substack{\vec{n}\\ n_j\geq 2}} \hat{Q}_Z^{(\vec{n})} 
\leq \sum_{\substack{\vec{n}\\ n_j\geq 2}} \alpha \hat{Q}_{\bot_j, Z}^{(\vec{n})}
\leq \sum_{\substack{\vec{n}}} \alpha \hat{Q}_{\bot_j, Z}^{(\vec{n})}
\label{QZnAupperQCKA}
\ee
holds. 
By defining 
\be
\hat{Q}_Z^{(n_j \geq 2)} \coloneqq \sum_{\substack{\vec{n} \\ n_j \ge 2}} \hat{Q}_Z^{(\vec{n})} ,
~~ \hat{Q}_{\bot_j, Z} \coloneqq \sum_{\substack{\vec{n}}} \hat{Q}_{\bot_j, Z}^{(\vec{n})}
\ee
for all $j\in \{0,\ldots,N\}$, 
Eq.~(\ref{QZnAupperQCKA}) is 
\be
\hat{Q}_Z^{(n_j \geq 2)} \leq \alpha \hat{Q}_{\bot_j, Z}.
\label{operelaQCKA}
\ee
By applying Azuma's inequality (see Appendix \ref{Azuma}), we obtain 
\be
\label{QZnAupperlastQCKA}
Q_Z^{(n_j \geq 2)} \lesssim \alpha Q_{\bot_j, Z}
\ee
in the asymptotic limit. 
From Eqs.~(\ref{nAnBupperQCKA}) and (\ref{QZnAupperlastQCKA}), we obtain 
\be
\sum_{\substack{ \vec{n}\\ \bigvee_{i=0}^{N}(n_i\geq2) }} Q_Z^{(\vec{n})} 
\lesssim 
\alpha \sum\limits_{i=0}^{N} Q_{\bot_i, Z}.
\ee
\end{proof}

Here, we define 
\be
q'_Z \coloneqq 
\sum_{\substack{ \vec{n}\\ \vec{n} \in \{0,1\}^{N+1}}}
 Q_Z^{(\vec{n})} 
 = Q_Z - \sum_{\substack{ \vec{n}\\ \bigvee_{i=0}^{N}(n_i\geq2) }} Q_Z^{(\vec{n})}
\label{qzQCKA}
\ee
for convenience. 
The following lemma gives a simplified upper bound on $Q_Z f_{\rm PA}$ by exploiting the concavity of the binary entropy function and generalizes Lemma~\ref{lemmaconcavity}.

\begin{lemma}
\label{lemmaconcavityQCKA}
Suppose $\bar{p}_Z \geq \bar{p}_X$ 
and 
\be
\frac{\bar{p}_Z^{N+1} \gamma}{\bar{p}_X^{N+1}} \frac{E_X}{q'_Z} \leq 1/2.
\ee
Then the following inequality holds:
\be
\begin{split}
&\sum_{\substack{ \vec{n}\\ \vec{n} \in \{0,1\}^{N+1},~\vec{n}\neq (1,..,1)}}
 Q_Z^{(\vec{n})} + Q_Z^{(1,..,1)} h\left(\frac{\bar{p}_Z^{N+1} \gamma}{\bar{p}_X^{N+1}} \frac{E_X^{(1,..,1)}}{Q_Z^{(1,..,1)}}\right) 
 \\
 &~~~~~~~~~~~~~~~~~~~~~~~~~~~~~~~~~~~~~~~~~~~~~~\lesssim q'_Z h \left( \frac{\bar{p}_Z^{N+1} \gamma}{\bar{p}_X^{N+1}} \frac{E_X}{q'_Z} \right). 
\label{eqlemma2QCKA}
\end{split}
\ee
\end{lemma}

\begin{proof}
The quantity $E_X^{(1,..,1)}$ can be upper bounded as follows:
\begin{align}
\bs
E_X^{(1,..,1)} &= E_X - \sum_{\substack{ \vec{n}\\ \vec{n} \in \{0,1\}^{N+1},~\vec{n}\neq (1,..,1)}} E_X^{(\vec{n})} 
- \sum_{\substack{ \vec{n}\\ \bigvee_{i=0}^{N}(n_i\geq2) }} E_X^{(\vec{n})} \\
& \leq E_X - \sum_{\substack{ \vec{n}\\ \vec{n} \in \{0,1\}^{N+1},~\vec{n}\neq (1,..,1)}} E_X^{(\vec{n})}.
\es
\label{XerrorineqQCKA}
\end{align}
For $\vec{n} \in\{0,1\}^{N+1}$ satisfying $\vec{n}\neq (1,..,1)$, let $\hat{E}_X^{(\vec{n})}$ be POVM elements corresponding to the event that $\{W_i = X\}_{i=0}^N, \bigoplus_{i=0}^N b_i =1$ and the observed photon numbers are given by $\vec n$. This is given by 
\be
\begin{split}
\hat{E}_X^{(\vec{n})} &= \frac{1}{2}(1-d)^{2(N+1)} 
\\ &~~~~~~~~\bigotimes_{i=0}^{N} \left([1-(1-d)^2]^2 \ket{\rm vac} \bra{\rm vac}_{i} \right)^{1-n_i}
\left( \bar{p}_X \hat{\mathbbm{1}}_i^{(1)}\right)^{n_i}, 
\end{split}
\ee
where the prefactor 1/2 implies the randomness of outcomes from dark counts. 
From Eq.~(\ref{hatQnQCKA}), we obtain 
\be
\hat{E}_X^{(\vec{n})} = \frac{1}{2 \gamma} 
\left( \prod_{i=0}^{N} \frac{\bar{p}_X^{n_i}}{\bar{p}_Z^{n_i}}\right) \hat{Q}_Z^{(\vec{n})}
\geq \frac{1}{2 \gamma}  \frac{\bar{p}_X^N}{\bar{p}_Z^N} \hat{Q}_Z^{(\vec{n})}
\ee
for $\vec{n} \in\{0,1\}^{N+1}$ satisfying $\vec{n}\neq (1,..,1)$ where we used $\bar{p}_Z \geq \bar{p}_X$ from the assumption of the lemma. 
By applying Azuma's inequality (see Appendix \ref{Azuma}), 
\be
E_X^{(\vec{n})} \gtrsim \frac{1}{2 \gamma}  \frac{\bar{p}_X^N}{\bar{p}_Z^N} Q_Z^{(\vec{n})}
\label{EXandQXQCKA}
\ee
holds for $\vec{n} \in\{0,1\}^{N+1}$ satisfying $\vec{n}\neq (1,..,1)$ in the asymptotic limit . 
By combining Eqs.~(\ref{XerrorineqQCKA}) and (\ref{EXandQXQCKA}), we obtain 
\begin{align}
\bs
\frac{{\bar{p}_Z}^{N+1}}{{\bar{p}_X}^{N+1}} E_X^{(1,..,1)} 
&\lesssim \frac{{\bar{p}_Z}^{N+1}}{{\bar{p}_X}^{N+1}} \left(E_X - \frac{1}{2 \gamma} \frac{\bar{p}_X^N}{\bar{p}_Z^N}
\sum_{\substack{ \vec{n}\\ \vec{n} \in \{0,1\}^{N+1},~\vec{n}\neq (1,..,1)}} Q_Z^{(\vec{n})}
\right) \\
& = \frac{{\bar{p}_Z}^{N+1}}{{\bar{p}_X}^{N+1}} E_X 
- \frac{1}{2 \gamma} \frac{\bar{p}_Z}{\bar{p}_X}  
\sum_{\substack{ \vec{n}\\ \vec{n} \in \{0,1\}^{N+1},~\vec{n}\neq (1,..,1)}} Q_Z^{(\vec{n})}.
\es
\label{preconcavQCKA}
\end{align}
From these considerations, we obtain  
\begin{align}
\bs
&\left( \sum_{\substack{ \vec{n}\\ \vec{n} \in \{0,1\}^{N+1},~\vec{n}\neq (1,..,1)}} Q_Z^{(\vec{n})} +Q_Z^{(1,..,1)} h\left( \frac{\bar{p}_Z^{N+1} \gamma}{\bar{p}_X^{N+1}} \frac{E_X^{(1,..,1)}}{Q_Z^{(1,..,1)}}\right) \right) /q'_Z \\
&=  \sum_{\substack{ \vec{n}\\ \vec{n} \in \{0,1\}^{N+1},~\vec{n}\neq (1,..,1)}} \frac{Q_Z^{(\vec{n})}}{q'_Z} h\left( \frac{1}{2} \right) + \frac{Q_Z^{(1,..,1)}}{q'_Z} h\left( \frac{\bar{p}_Z^{N+1} \gamma}{\bar{p}_X^{N+1}} \frac{E_X^{(1,..,1)}}{Q_Z^{(1,..,1)}}\right) \\
& \leq h\left(\sum_{\substack{ \vec{n}\\ \vec{n} \in \{0,1\}^{N+1},~\vec{n}\neq (1,..,1)}} \frac{Q_Z^{(\vec{n})}}{2q'_Z} + \frac{\bar{p}_Z^{N+1} \gamma}{\bar{p}_X^{N+1}} \frac{E_X^{(1,..,1)}}{q'_Z} \right) \\
& \lesssim h\left( \frac{{\bar{p}_Z}^{N+1} \gamma}{{\bar{p}_X}^{N+1}} \frac{E_X}{q'_Z} + \frac{1}{2q'_Z} \left( 1 - \frac{\bar{p}_Z} {\bar{p}_X} \right) 
\sum_{\substack{ \vec{n}\\ \vec{n} \in \{0,1\}^{N+1},~\vec{n}\neq (1,..,1)}} Q_Z^{(\vec{n})}
\right)
\\
& \leq h\Bigg( \frac{{\bar{p}_Z}^{N+1} \gamma}{{\bar{p}_X}^{N+1}} \frac{E_X}{q'_Z}\Bigg),
\es
\end{align}
where the first inequality follows from the concavity of the function $h(x)$, while the second inequality follows from Eq.~(\ref{preconcavQCKA}). 
The second and third inequalities use the fact that $h(x)$ is a non-decreasing continuous function in the range of $0\leq x \leq 1/2$. 
From the assumption of the lemma, the terms in the brackets are no greater than 1/2.  
The third inequality uses $\bar{p}_Z / \bar{p}_X \geq 1$ from the assumption of the lemma. 
This result leads to Eq.~(\ref{eqlemma2QCKA}).
\end{proof}

By combining the three lemmas, we obtain the following main theorem for the passive QCKA. 
\begin{thm} 
\label{Theo:QCKA}
If the relations
\be
\bar{p}_Z \geq \bar{p}_X, ~~
\frac{{\bar{p}_Z}^2 \gamma}{{\bar{p}_X}^2} \frac{E_X}{Q_Z -  \alpha \sum\limits_{i=0}^{N} Q_{\bot_i, Z} } \leq \frac{1}{2} 
\label{pbarineqQCKA} 
\ee 
hold,  
the lower bound on the secure key rate $R$ per round of the passive QCKA in the asymptotic limit is given by 
\begin{align}
\begin{split}
&R \gtrsim \left(Q_Z - \alpha \sum\limits_{i=0}^{N} Q_{\bot_i, Z}\right)\left(1-h\left(\frac{({\bar{p}_Z}^{N+1} \gamma/{\bar{p}_X}^{N+1})E_X}{Q_Z - \alpha \sum\limits_{i=0}^{N} Q_{\bot_i, Z}}\right)  \right) - Q_Z f_{\rm EC}. 
\end{split}
\label{finalkeyrateQCKA}
\end{align}
\end{thm}

\begin{proof} 
Applying Lemma \ref{lemmasingleQCKA} to Eq.~(\ref{PAamountQCKA}), we first obtain 
\begin{align}
\bs
Q_Z f_{\rm PA} 
& \sim \sum_{\substack{ \vec{n} \\ \bigvee_{i=0}^{N}(n_i\geq2) }} Q_Z^{(\vec{n})} 
+ \sum_{\substack{\vec{n} \\ \vec{n} \in \{0,1\}^{N+1},~\vec{n}\neq (1,..,1) }} Q_Z^{(\vec{n})}\\
&~~~~~~~~~~~~~~~~~~~~~~~~~~~~~~~~~~+Q_Z^{(1,..,1)} h\left(\frac{{\bar{p}_Z}^{N+1} \gamma}{{\bar{p}_X}^{N+1}} \frac{E_X^{(1,..,1)}}{Q_Z^{(1,..,1)}}\right) .
\label{PAboundQCKA}
\es
\end{align}
Applying Lemma \ref{lemmaconcavityQCKA} to Eq.~(\ref{PAboundQCKA}), we further obtain 
\be
\begin{split}
Q_Z f_{\rm PA} &\lesssim \sum_{\substack{ \vec{n} \\ \bigvee_{i=0}^{N}(n_i\geq2) }} Q_Z^{(\vec{n})} + q'_Z \left(\frac{\bar{p}_Z^{N+1} \gamma}{\bar{p}_X^{N+1}} \frac{E_X}{q'_Z} \right) \\ 
&= Q_Z - q'_Z + q'_Z h \left(\frac{\bar{p}_Z^{N+1} \gamma}{\bar{p}_X^{N+1}} \frac{E_X}{q'_Z} \right),
\end{split}
\ee
if $\bar{p}_Z \geq \bar{p}_X$ and 
\be
\frac{\bar{p}_Z^{N+1} \gamma}{\bar{p}_X^{N+1}} \frac{E_X}{q'_Z} \leq 1/2 \label{assulemma2QCKA}
\ee
is satisfied. 
The secure key rate per round is then given by 
\be
R = Q_Z(1 - f_{\rm PA} - f_{\rm EC}) \gtrsim q'_Z \left(1- h \left( \frac{\bar{p}_Z^{N+1} \gamma}{\bar{p}_X^{N+1}} \frac{E_X}{q'_Z} \right) \right) - Q_Zf_{\rm EC}. 
\label{RqzQCKA}
\ee
From Eq.~(\ref{qzQCKA}) and Lemma \ref{lemmamultiQCKA}, we have 
\be
q'_Z = Q_Z - \sum_{\substack{ \vec{n} \\ \bigvee_{i=0}^{N}(n_i\geq2) }} Q_Z^{(\vec{n})}  
\gtrsim Q_Z - \alpha \sum\limits_{i=0}^{N} Q_{\bot_i, Z}. \label{qZlowerQCKA}
\ee
From this bound, Eq.~(\ref{assulemma2QCKA}) is satisfied if 
\be
\frac{{\bar{p}_Z}^2 \gamma}{{\bar{p}_X}^2} \frac{E_X}{Q_Z - \alpha \sum\limits_{i=0}^{N} Q_{\bot_i, Z}} \leq \frac{1}{2} \label{keyconditionQCKA}
\ee 
holds. 
Applying the bound Eq.~(\ref{qZlowerQCKA}) to Eq.~(\ref{RqzQCKA}), the secure key rate is lower bounded by  
\begin{align}
\begin{split}
&R \gtrsim \left(Q_Z - \alpha \sum\limits_{i=0}^{N} Q_{\bot_i, Z}\right)\left(1-h\left(\frac{({\bar{p}_Z}^{N+1} \gamma/{\bar{p}_X}^{N+1})E_X}{Q_Z - \alpha \sum\limits_{i=0}^{N} Q_{\bot_i, Z}}\right)  \right) - Q_Z f_{\rm EC}, 
\end{split}
\end{align}
if Eq.~(\ref{keyconditionQCKA}) is satisfied. 
\end{proof}

\section{Conclusion} \label{Sec:Conclusion}
In this paper, we established the security of passive entanglement-based key distribution protocols in the asymptotic regime. We first proved the security of the passive BBM92 protocol with biased basis choice, for which conventional proof techniques are not directly applicable. Our approach is based on the introduction of a virtual-qubit description for entanglement-based protocols, which enables the security analysis to be reduced to that of the BB84 protocol. 
We also derived analytical expressions for the observed quantities in the passive BBM92 protocol and used them to evaluate its performance. Our numerical simulations showed that the key generation rate of the passive protocol is almost identical to that of the active protocol, indicating that passive implementations can achieve the same level of performance while substantially simplifying the measurement setup. 
Moreover, we generalized the analytical framework developed for the passive BBM92 protocol to the multipartite setting and proved the security of passive QCKA.

Although our analysis was presented for polarization-encoded protocols, the proposed framework is also applicable to time-bin encoding. Since the basis-dependent asymmetry inherent in time-bin and phase measurements can be naturally incorporated into our model, our results provide a security proof for an important class of practical implementations for which passive basis choice is particularly attractive.

Finally, our analysis assumes the asymptotic limit and does not account for finite-size effects. Establishing finite-key security using concentration inequalities, such as Kato's inequality~\cite{2020Kato}, is an important next step and will further strengthen the practical relevance of passive entanglement-based cryptographic protocols.

{\it Note added}: After completion of this work, we became aware of a recent preprint analyzing cross-click operators and their possible application to the passive BBM92 protocol~\cite{2026Ye}. While that work derives a lower bound on low-photon-number block weight for the passive BBM92 protocol using cross click events, to the best of our knowledge, it does not itself provide a complete security proof.

\section*{Acknowledgment}
We thank Akihiro Mizutani and Takuya Ikuta for helpful discussions. 
This work was supported by `MIC R\&D of ICT Priority Technology (JPMI00316)' of Ministry of Internal Affairs and Communications of Japan.

\section*{Appendix}
\appendix

\section{Azuma's inequality} \label{Azuma}
Here, we show the application of Azuma's inequality in the asymptotic limit of $m_{\rm rep} \to \infty$. 
Let $\hat{F}(\Omega = \omega)$ denote a POVM element corresponding to the outcome of $\Omega = \omega$. 
We specifically show that if 
\be
\hat{F}(\Omega = \omega) = a \hat{F}(\Omega = \omega')
\ee
holds for a constant $a>0$, then 
\be
\frac{m(\Omega = \omega)}{m_{\rm rep}} = \frac{a m(\Omega = \omega')}{m_{\rm rep}}
\ee
holds in the asymptotic limit. 
Similarly, if 
\be
\hat{F}(\Omega = \omega) \leq a \hat{F}(\Omega = \omega')
\ee
holds, then 
\be
\frac{m(\Omega = \omega)}{m_{\rm rep}} \leq \frac{a m(\Omega = \omega')}{m_{\rm rep}}
\ee
holds in the asymptotic limit. 
Let $b_i(\Omega = \omega) \in \{0,1\}$ be a random variable for $i$-th round, satisfying $b_i(\Omega = \omega) = 1$ if $\Omega = \omega$ and otherwise 
$b_i(\Omega = \omega) = 0$. Let $p_i(\Omega = \omega)$ be a probability for $i$-th round of obtaining the outcome $\Omega = \omega$. It is given by 
\be
\label{piOmega}
p_i(\Omega = \omega) = {\rm tr} \left(\hat{F}(\Omega = \omega) \hat{\rho}_{AB|\mathcal{F}_{i-1}}(i)\right),
\ee
where $\hat{\rho}_{AB|\mathcal{F}_{i-1}}(i)$ is 
a quantum state of system $AB$ in the $i$-th round conditioned on $\mathcal{F}_{i-1}$, which contains all classical information up to the $(i-1)$-th round. 
We sum up these elements up to the $k$-th round and define 
\be
M_k(\Omega = \omega) \coloneqq \sum\limits_{i=1}^k b_i(\Omega = \omega), ~~ P_k(\Omega = \omega) \coloneqq \sum\limits_{i=1}^k p_i(\Omega = \omega). 
\ee
Note that $M_{m_{\rm rep}}(\Omega = \omega)$ equals $m(\Omega = \omega)$. 
We also define their difference 
\be
\Delta_k^{(\Omega = \omega)}  \coloneqq M_k(\Omega = \omega) - P_k(\Omega = \omega)
\ee
for $1\leq k \leq m_{\rm rep}$ and $\Delta_0^{(\Omega = \omega)} = 0$.
Since 
\be
\Delta_k^{(\Omega = \omega)} - \Delta_{k-1}^{(\Omega = \omega)} = b_k(\Omega = \omega) - p_k(\Omega = \omega)
\ee
holds for $1\leq k \leq m_{\rm rep}$, we obtain 
\be
E[\Delta_k^{(\Omega = \omega)} - \Delta_{k-1}^{(\Omega = \omega)} | \mathcal{F}_{k-1}] = E[b_k(\Omega = \omega) | \mathcal{F}_{k-1}] - p_k(\Omega = \omega) = 0,
\label{martin1}
\ee
where $E[X|Y]$ denotes the expectation of $X$ condition on $Y$. 
On the other hand, we have 
\be
\begin{split}
E[\Delta_k^{(\Omega = \omega)} - \Delta_{k-1}^{(\Omega = \omega)} | \mathcal{F}_{k-1}] &= E[\Delta_k^{(\Omega = \omega)} | \mathcal{F}_{k-1}] - E[\Delta_{k-1}^{(\Omega = \omega)} | \mathcal{F}_{k-1}] \\
& =  E[\Delta_k^{(\Omega = \omega)} | \mathcal{F}_{k-1}] - \Delta_{k-1}^{(\Omega = \omega)}. 
\label{martin2}
\end{split}
\ee
Eqs.~(\ref{martin1}) and (\ref{martin2}) lead to 
\be
E[\Delta_k^{(\Omega = \omega)} | \mathcal{F}_{k-1}] = \Delta_{k-1}^{(\Omega = \omega)}, 
\ee
and thus $\{\Delta_k^{(\Omega = \omega)}\}$ satisfies Martingale condition. 
Moreover, $\{\Delta_k^{(\Omega = \omega)}\}$ satisfies the bounded difference condition, that is, $|\Delta_k^{(\Omega = \omega)} - \Delta_{k-1}^{(\Omega = \omega)}| \leq 1$ for any $ 1\leq k \leq m_{\rm rep}$. 
Therefore, Azuma's inequality can be applied as follows: for $\epsilon >0$, 
\be
{\rm Pr}(|\Delta_{m_{\rm rep}}^{(\Omega = \omega)}| \geq m_{\rm rep} \epsilon) \leq 2 e^{-m_{\rm rep} \epsilon^2 /2}. 
\ee
In other words, by using the relation $M_{m_{\rm rep}}(\Omega = \omega) = m(\Omega = \omega)$, 
\be
|m(\Omega = \omega) - P_{m_{\rm rep}}(\Omega = \omega) | < m_{\rm rep} \epsilon
\label{azuma1}
\ee
holds except with probability $2e^{-m_{\rm rep} \epsilon^2 /2}$. 

Here, suppose that
\be
\hat{F}(\Omega = \omega) = a \hat{F}(\Omega = \omega')
\ee
holds for $a >0$. 
From Eq.~(\ref{piOmega}), we obtain $p_i(\Omega = \omega) = a p_i(\Omega = \omega')$ for any $i$, which leads to 
\be
P_{m_{\rm rep}}(\Omega = \omega) = a  P_{m_{\rm rep}}(\Omega = \omega'). 
\ee
From Eq.~(\ref{azuma1}), the relation
\be
|m(\Omega = \omega) - a m(\Omega = \omega')| < (1+ a) m_{\rm rep} \epsilon, \label{MM}
\ee
namely, 
\be
\left| \frac{m(\Omega = \omega)}{m_{\rm rep}} - a \frac{m(\Omega = \omega')}{m_{\rm rep}} \right| < (1+ a)  \epsilon, \label{MM}
\ee
holds except with probability less than $4e^{-m_{\rm rep} \epsilon^2 /2}$, which converges to 0 for arbitrary $\epsilon$ in the asymptotic limit of $m_{\rm rep} \to \infty$. 
Therefore, $m(\Omega = \omega)/m_{\rm rep}$ equals $a m(\Omega = \omega)/m_{\rm rep}$ in the asymptotic limit. 
In a similar vein, if $\hat{F}(\Omega = \omega) \leq a \hat{F}(\Omega = \omega')$ holds, we obtain 
$m(\Omega = \omega)/m_{\rm rep} \leq a m(\Omega = \omega)/m_{\rm rep}$ in the asymptotic limit.

\section{POVM elements of single-photon detection} \label{singlePOVM}
Here, we derive the POVM elements associated with single-photon detection events. 
We define $Z, X$ as systems of Alice's $Z$ line and $X$ line in Fig.~\ref{setup2}, respectively. 
We  define ($\hat{z}^{\dagger}_H$, $\hat{z}^{\dagger}_V$) and ($\hat{x}^{\dagger}_H$, $\hat{x}^{\dagger}_V$) as sets of 
creation operators of horizontal and vertical polarizations in systems $Z$ and $X$, respectively. 
We also define $\hat{x}^{\dagger}_D := (\hat{x}^{\dagger}_H + \hat{x}^{\dagger}_V)/\sqrt{2}$ and $\hat{x}^{\dagger}_{\bar{D}} := (\hat{x}^{\dagger}_H - \hat{x}^{\dagger}_V)/\sqrt{2}$. With a notation 
\begin{align}
\hat{N}_{\xi}^{(1)} &\coloneqq \hat{\xi}^{\dagger} \ket{\rm vac} \bra{\rm vac}_{ZX} \hat{\xi},
\end{align}
for $\xi \in \{z_H, z_V, x_D, x_{\bar{D}}\}$, the POVM elements in system $ZX$ under the condition that the outcome of QND measurement is $n_A =1$ 
are described as follows: 
\begin{equation}
\begin{split}
&\hat{F}_{Z0}^{(1)} =  \hat{N}_{z_H}^{(1)}  (1-d)^3,\\
&\hat{F}_{Z1}^{(1)} =  \hat{N}_{z_V}^{(1)}  (1-d)^3,  \\
&\hat{F}_{Z, \rm double}^{(1)} =  d ( \hat{N}_{z_H}^{(1)} +  \hat{N}_{z_V}^{(1)})  (1-d)^2, \\
&\hat{F}_{X0}^{(1)} =  \hat{N}_{x_D}^{(1)}  (1-d)^3,\\
&\hat{F}_{X1}^{(1)} =  \hat{N}_{x_{\bar{D}}}^{(1)}  (1-d)^3,  \\
&\hat{F}_{X, \rm double}^{(1)} =  d ( \hat{N}_{x_D}^{(1)} +  \hat{N}_{x_{\bar{D}}}^{(1)})  (1-d)^2, \\
&\hat{F}_{\bot}^{(1)} = \left(1-(1-d)^2\right) \hat{\mathbbm{1}}_{ZX}^{(1)}, \\  
\label{POVMoriginals}
\end{split}
\end{equation}
with $\hat{F}_{W'b}^{(1)}$ corresponding to the outcome $b\in\{0, 1\}$ in $W'\in\{Z, X\}$ line, 
$\hat{F}_{W', \rm double}^{(1)}$ corresponding to the double-click events in $W'\in\{Z, X\}$ line and $\hat{F}_{\bot}^{(1)}$ corresponding to cross-click events. 
Note that the sum of the above POVM elements equals the projection onto the single-photon subspace of system $ZX$
\begin{align}
\hat{\mathbbm{1}}_{ZX}^{(1)} \coloneqq \hat{N}_{z_H}^{(1)} +  \hat{N}_{z_V}^{(1)} +  \hat{N}_{x_D}^{(1)} +  \hat{N}_{x_{\bar{D}}}^{(1)}.
\end{align} 

Now we rewrite those POVM elements as operators on system $A$ instead of system $ZX$. 
An ancillary system $E$ is introduced so that the dimension of $AE$ equals that of $ZX$. 
We define $\hat{a}_H^{\dagger}$ and $\hat{a}_V^{\dagger}$ be creation operators in Alice's system $A$ such that 
\begin{align}
\ket{H}_A = \hat{a}_H^{\dagger} \ket{\rm{vac}}_A,~~\ket{V}_A = \hat{a}_V^{\dagger} \ket{\rm{vac}}_A
\end{align}
holds. 
Let $\hat{U}$ be a unitary operator from $AE$ to $ZX$ representing evolution of the beam splitter in Fig.~\ref{setup2}, whose action is  
\begin{align}
\begin{split}
\hat{U} \hat{a}^{\dagger}_H \hat{U}^{\dagger} = \sqrt{\bar{p}_Z} \hat{z}^{\dagger}_H + \sqrt{\bar{p}_X} \hat{x}^{\dagger}_H,~~ 
\hat{U} \hat{a}^{\dagger}_V \hat{U}^{\dagger} = \sqrt{\bar{p}_Z} \hat{z}^{\dagger}_V+ \sqrt{\bar{p}_X} \hat{x}^{\dagger}_V. \label{HVevolution}
\end{split}
\end{align}
Here we define a projection operator $\hat{M}_{A,H}^{(1)}$ from the single-photon subspace on system $A$ onto $\hat{z}^{\dagger}_H \ket{\rm vac}_{ZX}$, which is described as $\hat{M}_{A,H}^{(1)} =  {}_{ZX}\bra{\rm vac} \hat{z}_H \hat{U} \ket{\rm vac}_{E}$. 
By denoting $\hat{M}_{A,H}^{(1)}  \ket{H}_A \eqqcolon \gamma_{HH}$ and $\hat{M}_{A,H}^{(1)} \ket{V}_A \eqqcolon \gamma_{HV}$, 
\be
\hat{M}_{A,H}^{(1)} = \hat{M}_{A,H}^{(1)} \hat{\mathbbm{1}}_A^{(1)}  = {}_A \bra{H} \gamma_{HH} +  {}_A \bra{V} \gamma_{HV}
\label{MBH}
\ee
holds, where we defined $\hat{\mathbbm{1}}_A^{(1)} \coloneqq \ket{H} \bra{H}_A + \ket{V} \bra{V}_A$. Now we calculate $\gamma_{HH}$ and $\gamma_{HV}$ as follows: 
\begin{equation}
\begin{split}
\gamma_{HH} &= {}_{ZX}\bra{\rm vac} \hat{z}_H \hat{U} \hat{a}^{\dagger}_H \ket{\rm vac}_{AE} \\
&= {}_{ZX}\bra{\rm vac} \hat{z}_H \hat{U} \hat{a}^{\dagger}_H \hat{U}^{\dagger}  \ket{\rm vac}_{ZX} \\
&= {}_{ZX}\bra{\rm vac} \hat{z}_H(\sqrt{\bar{p}_Z} \hat{z}^{\dagger}_H + \sqrt{\bar{p}_X} \hat{x}^{\dagger}_H )\ket{\rm vac}_{ZX}  \\
&= \sqrt{\bar{p}_Z}, \\
\gamma_{HV}
&= {}_{ZX}\bra{\rm vac} \hat{z}_H(\sqrt{\bar{p}_Z} \hat{z}^{\dagger}_V + \sqrt{\bar{p}_X} \hat{x}^{\dagger}_V )\ket{\rm vac}_{ZX}  \\
&= 0. 
\end{split}
\end{equation}
Therefore, from Eq. (\ref{MBH}), we obtain $\hat{M}_{A,H}^{(1)} =   {}_A\bra{H} \sqrt{\bar{p}_Z}$, and thus the POVM element on system $A$ is 
${\hat{M}_{A,H}^{(1) \dagger}} \hat{M}_{A,H}^{(1)} = \bar{p}_Z  \ket{H} \bra{H}_A$. 
Similarly, we have the following relations: 
\begin{equation}
\begin{split}
&\hat{N}_{z_H}^{(1)} \to \bar{p}_Z \ket{H}\bra{H}_A, ~~
\hat{N}_{z_V}^{(1)}  \to \bar{p}_Z \ket{V}\bra{V}_A, \\
&\hat{N}_{x_D}^{(1)}  \to \bar{p}_X \ket{D}\bra{D}_A,  ~~
\hat{N}_{x_{\bar{D}}}^{(1)}   \to \bar{p}_X \ket{\bar{D}}\bra{\bar{D}}_A.
\end{split}
\end{equation}
From these results, POVM elements in Eq.~(\ref{POVMoriginals}) are represented with operators on system $A$ as follows:
\begin{equation}
\begin{split}
&\hat{F'}_{Z_A0}^{(1)} =  \bar{p}_Z  \ket{H} \bra{H}_A  (1-d)^3, \\
&\hat{F'}_{Z_A1}^{(1)} =  \bar{p}_Z  \ket{V} \bra{V}_A   (1-d)^3,  \\
&\hat{F'}_{Z_A, \rm double}^{(1)} = \bar{p}_Z d  (1-d)^2 \hat{\mathbbm{1}}_A^{(1)},  \\
&\hat{F'}_{X_A0}^{(1)} =  \bar{p}_X  \ket{D} \bra{D}_A  (1-d)^3, \\
&\hat{F'}_{X_A1}^{(1)} = \bar{p}_X  \ket{\bar{D}} \bra{\bar{D}}_A  (1-d)^3,  \\
&\hat{F'}_{X_A, \rm double}^{(1)} = \bar{p}_X d  (1-d)^2 \hat{\mathbbm{1}}_A^{(1)},  \\
&\hat{F'}_{\bot_A}^{(1)} = \left( 1-(1-d)^2\right) \hat{\mathbbm{1}}_A^{(1)}.
\end{split}
\end{equation}
Note that the sum of those elements equals $ \hat{\mathbbm{1}}_A^{(1)}$.

\section{POVM elements of $n_A$-photon detection} \label{multiPOVM}
Here, we derive the POVM elements associated with $n_A$-photon detection events. 
Let us define the following Fock states on system $A$ and system $ZX$:
 \begin{align}
& \ket{\vec{m}_A}_A = \ket{m_H, m_V}_A \coloneqq \frac{(\hat{a}_H^{\dagger})^{m_H} (\hat{a}_V^{\dagger})^{m_V}}{\sqrt{m_H! m_V!}} \ket{\rm vac}_A, \label{mHmV}\\
& \ket{\vec{m}_{ZX}}_{ZX}  = \ket{m_{ZH}, m_{ZV},  m_{XH},  m_{XV}}_{ZX}  \nonumber \\
 &\coloneqq \frac{(\hat{z}_H^{\dagger})^{m_{ZH}} (\hat{z}_V^{\dagger})^{m_{ZV}} (\hat{x}_H^{\dagger})^{m_{XH}} (\hat{x}_V^{\dagger})^{m_{XV}}}{\sqrt{m_{ZH}! m_{ZV}!  m_{XH}!  m_{XV}!  }}   \ket{\rm vac}_{ZX}.  \label{ZXLbasis}
 \end{align}
 We denote the following three POVM elements for $n_A\geq 1$ on system $Z X $: 
 an element $\hat{F}_W^{(n_A)}$ corresponding to key generation in $W \in \{Z, X\}$ basis and 
 an element $\hat{F}_{\bot}^{(n_A)}$ corresponding to cross-click events.
By using the basis states in Eq. (\ref{ZXLbasis}) with a notation $P(\ket{\cdot}) \coloneqq \ket{\cdot}\bra{\cdot}$, the element $\hat{F}_Z^{(n_A)}$ is described as
\begin{equation}
\begin{split}
&\hat{F}_Z^{(n_A)} = 
 (1-d)^2\bigg(  \sum\limits_{m_{H} = 0}^{n_A}  P(\ket{m_{H}, n_A - m_{H},  0, 0}_{ZX})  \bigg).   \label{FZnorg} 
\end{split}
\end{equation}
Similarly, the other elements are described as follows:
\begin{align}
&\hat{F}_X^{(n_A)} = (1-d)^2\bigg(  \sum\limits_{m_{H} = 0}^{n_A}  P(\ket{0, 0, m_{H}, n_A - m_{H}}_{ZX})   \bigg),   \label{FXnorg} \\
&\hat{F}_{\bot}^{(n_A)} = \hat{\mathbbm{1}}_{ZX}^{(n_A)} - \hat{F}_Z^{(n_A)} -\hat{F}_X^{(n_A)} , \label{Fcrossnorg}
\end{align}
where 
\begin{align}
\hat{\mathbbm{1}}_{ZX}^{(n_A)} \coloneqq \sum_{\substack{\vec{m}_{ZX} \\ m _{ZH}+m_{ZV}+m_{XH}+m_{XV}=n_A}} 
P(\ket{\vec{m}_{ZX}}_{ZX}). 
\end{align}
Again, we introduce an ancillary system $E$ so that the dimension of $AE$ equals that of system $ZX$. Let $\hat{\mathbbm{1}}_{AE}^{(n_A)}$ be an projection operator on $n_A$-photon subspace of system $AE$. 
Using the unitary operator $\hat{U}$ representing the beam splitter, 
\be
\hat{\mathbbm{1}}_{AE}^{(n_A)} = \hat{U}^{\dagger} \hat{\mathbbm{1}}_{ZX}^{(n_A)} \hat{U}
\label{identityunitary}
\ee
holds. Here we define 
\be
\hat{\mathbbm{1}}_{A}^{(n_A)} \coloneqq \sum_{\substack{\vec{m}_A \\  m_H + m_V=n_A}} 
P(\ket{\vec{m}_A}_{A}), \label{identityonnb}
\ee
and we then have 
\be
{}_{E} \bra{\rm vac}  \hat{U}^{\dagger} \hat{\mathbbm{1}}_{ZX}^{(n_A)} \hat{U} \ket{\rm vac}_{E} =\hat{\mathbbm{1}}_{A}^{(n_A)}. \label{identityreduction}
\ee

Now we represent $\hat{F}_Z^{(n_A)}$ in Eq. (\ref{FZnorg}) in terms of system $A$. 
Let us define
\be
\ket{\vec{m}_{Z}}_{ZX} \coloneqq \ket{m_{H}, n_A-m_{H},  0, 0}_{ZX}. \label{nZL}
\ee
A measurement operator  $\hat{M}_{\vec{m}_{Z}}: A \to \mathbbm{C}$ which corresponds to the projection onto $\ket{\vec{m}_{Z}}_{ZX}$ is  represented as 
$\hat{M}_{\vec{m}_{Z}}= {}_{ZX}\bra{\vec{m}_{Z}} \hat{U} \ket{\rm vac}_{E}$. 
For arbitrary 
\be
\ket{\vec{m}'_A}_A \coloneqq \ket{m'_H, m'_V}_A, 
\ee
let $\gamma_{\vec{m}'_A, \vec{m}_{Z}}  \coloneqq \hat{M}_{\vec{m}_{Z}} \ket{\vec{m}'_A}_A.$ Then we can write $\hat{M}_{\vec{m}_{Z}}$ as 
\begin{align}
\hat{M}_{\vec{m}_{Z}} =  \hat{M}_{\vec{m}_{Z}}  \sum\limits_{\vec{m}'_A} \ket{\vec{m}'_A}\bra{\vec{m}'_A}_A
=\sum\limits_{\vec{m}'_A}  {}_A\bra{\vec{m}'_A} \gamma_{\vec{m}'_A, \vec{m}_{Z}}. \label{Kraus_multi} 
\end{align}
By using Eq. (\ref{HVevolution}), we calculate $\gamma_{\vec{m}'_A, \vec{m}_{Z}}$ as follows:
\begin{widetext}
\begin{equation}
\begin{split}
\gamma_{\vec{m}'_A, \vec{m}_{Z}} &= {}_{ZX}\bra{\vec{m}_{Z}} \hat{U} \ket{\vec{m}'_A}_A \ket{\rm vac}_{E} 
= {}_{ZX}\bra{\vec{m}_{Z}} \frac{ \hat{U} (\hat{a}_H^{\dagger})^{m'_H} (\hat{a}_V^{\dagger})^{m'_V} } { \sqrt{m'_H! m'_V!}} \ket{\rm vac}_{AE}  
= {}_{ZX}\bra{\vec{m}_{Z}} \frac{ (\hat{U} \hat{a}_H^{\dagger} \hat{U}^{\dagger})^{m'_H} (\hat{U} \hat{a}_V^{\dagger} \hat{U}^{\dagger})^{m'_V} } { \sqrt{m'_H! m'_V!}} \ket{\rm vac}_{ZX}  \\
&= {}_{ZX}\bra{\vec{m}_{Z}} \frac{1 }{\sqrt{m'_H!m'_V!}}  \Big(\sqrt{\bar{p}_Z} \hat{z}_H^{\dagger} + \sqrt{\bar{p}_X} \hat{x}_H^{\dagger} \Big)^{m'_H}  \Big(\sqrt{\bar{p}_Z} \hat{z}_V^{\dagger} + \sqrt{\bar{p}_X} \hat{x}_V^{\dagger}  \Big)^{m'_V}  \ket{\rm vac}_{ZX} \\
& = {}_{ZX}\bra{\vec{m}_{Z}} \frac{1 }{\sqrt{m'_H!m'_V!}} 
 \sum\limits_{i=0}^{m'_H} \frac{m'_H!}{i! (m'_H - i)!} (\sqrt{\bar{p}_Z} \hat{z}_H^{\dagger})^i  (\sqrt{\bar{p}_X} \hat{x}_H^{\dagger})^{m'_H-i}  \sum\limits_{j=0}^{m'_V} \frac{m'_V!}{j! (m'_V -j)!} 
(\sqrt{\bar{p}_Z} \hat{z}_V^{\dagger})^{j}  (\sqrt{\bar{p}_X} \hat{x}_V^{\dagger})^{m'_V-j} \ket{\rm vac}_{ZX}  \\
&=  {}_{ZX}\bra{\vec{m}_{Z}} \sum\limits_{i=0}^{m'_H} \sum\limits_{j=0}^{m'_V} \sqrt{\frac{m'_H!}{i! (m'_H - i)!}  \frac{m'_V!}{j! (m'_V - j)!}}  
(\sqrt{\bar{p}_Z})^{i+j} (\sqrt{\bar{p}_X})^{m'_H + m'_V -i-j}  \ket{i, j, m'_H-i, m'_V-j}_{ZX}. \label{gammaprocess}
\end{split}
\end{equation}
\end{widetext}
From Eq. (\ref{nZL}), a term  in Eq. (\ref{gammaprocess}) is not zero only if $i = m_{H}, j = n_A-m_{H}, m'_H = m_H$ and $ m'_V = n_A - m_H$ hold. 
If $m'_H \neq m_H$ or $m'_V \neq n_A - m_H$, then $\gamma_{\vec{m}'_A, \vec{m}_{Z}} = 0$ and if $m'_H = m_H$ and $m'_V = n_A - m_H$, then 
\begin{equation}
\gamma_{\vec{m}'_A, \vec{m}_{Z}} = (\sqrt{\bar{p}_Z})^{n_A}.
\end{equation}
From Eq. (\ref{Kraus_multi}), we obtain 
\begin{align}
\begin{split}
\hat{M}_{\vec{m}_{Z}}  &=    (\sqrt{\bar{p}_Z})^{n_A} {}_A\bra{m_H, n_A - m_H}. 
\end{split}
\end{align}
Thus, a POVM element on system $A$ corresponding to projection onto $\ket{\vec{m}_{Z}}_{ZX}$ is 
\begin{align}
\begin{split}
& P({}_{E}\bra{\rm vac}\hat{U}^{\dagger} \ket{\vec{m}_{Z}}_{ZX})  
= \hat{M}_{\vec{m}_{Z}} ^{\dagger} \hat{M}_{\vec{m}_{Z}} \\
&=   \bar{p}_Z^{n_A} P(\ket{m_H, n_A - m_H}_A) \label{mouikutsu}. 
\end{split}
\end{align}
By using Eq. (\ref{mouikutsu}), the POVM element in Eq. (\ref{FZnorg}) is described in terms of system $A$ as follows:   
\begin{equation}
\begin{split}
 \hat{F'}_{Z_A}^{(n_A)} & \coloneqq
{}_{E} \bra{\rm vac} \hat{U}^{\dagger} \hat{F}_Z^{(n_A)} \hat{U} \ket{\rm vac}_{E}  \\
&= (1-d)^2\bigg(  \sum\limits_{m_H = 0}^{n_A} \bar{p}_Z^{n_A}  P(\ket{m_H, n_A - m_H}_A)  \bigg)  \\
&= (1-d)^2 \bar{p}_Z ^{n_A}  \hat{\mathbbm{1}}_{A}^{(n_A)}. \label{FZn}
\end{split}
\end{equation}
In similar ways, Eq. (\ref{FXnorg}) is described in terms of system $A$:
\begin{align}
\hat{F'}_{X_A}^{(n_A)} &\coloneqq 
{}_{E} \bra{\rm vac} \hat{U}^{\dagger} \hat{F}_X^{(n_A)}  \hat{U} \ket{\rm vac}_{E} \nonumber \\ 
&= (1-d)^2 \bar{p}_X^{n_A}  \hat{\mathbbm{1}}_{A}^{(n_A)}. \label{FXn}
\end{align}
By using 
Eqs. (\ref{identityreduction}), (\ref{FZn}) and (\ref{FXn}), we describe Eq. (\ref{Fcrossnorg}) in terms of system $A$:
\begin{equation}
\begin{split}
\hat{F'}_{\bot_A}^{(n_A)} &\coloneqq 
{}_{E} \bra{\rm vac} \hat{U}^{\dagger} \hat{F}_{\bot}^{(n_A)} \hat{U} \ket{\rm vac}_{E}  \\
&= \hat{\mathbbm{1}}_{A}^{(n_A)} - \hat{F'}_{Z_A}^{(n_A)} -\hat{F'}_{X_A}^{(n_A)} \\
&=\Big( 1 - (1-d)^2 \Big( \bar{p}_Z^{n_A} + \bar{p}_X^{n_A} \Big) \Big) \hat{\mathbbm{1}}_{A}^{(n_A)}.
\label{Fcrossn}
\end{split}
\end{equation}

\section{Proof of the basis invariance of the PDC source state}
\label{Appe:basisrotation}
We show that the state $\ket{\Phi_n}_{AB}$ retains the same form when expressed in the diagonal polarization basis, as given by Eq.~(\ref{Phiequiv}).
The Fock states can be written in terms of the creation operators as
\be
\ket{n-m,m}_A =
\frac{(a_H^\dagger)^{n-m}(a_V^\dagger)^m}
{\sqrt{(n-m)!m!}}
\ket{{\rm vac}}_A,
\ee
and similarly for Bob. Therefore,
\be
\ket{\Phi_n}_{AB} = 
\frac{1}{\sqrt{n+1}}
\sum_{m=0}^{n}
\frac{(-1)^m
(a_H^\dagger)^{n-m}(a_V^\dagger)^m
(b_H^\dagger)^m(b_V^\dagger)^{n-m}}
{(n-m)!m!}
\ket{{\rm vac}}_{AB},
\ee
where $b_H^\dagger$ and $b_V^\dagger$ are creation operators in Bob's system $B$. 
Introducing the operator
\be
J^\dagger \coloneqq
a_H^\dagger b_V^\dagger -
a_V^\dagger b_H^\dagger,
\ee
the binomial theorem gives
\be
(J^\dagger)^n = 
\sum_{m=0}^{n}
(-1)^m
\binom{n}{m}
(a_H^\dagger b_V^\dagger)^{n-m}
(a_V^\dagger b_H^\dagger)^m
\ket{{\rm vac}}_{AB}.
\ee
Since all creation operators commute with one another,
\be
(J^\dagger)^n \ket{{\rm vac}}_{AB}= 
n!
\sum_{m=0}^{n}
(-1)^m
\ket{n-m,m}_A
\ket{m,n-m}_B,
\ee
which yields
\be
\ket{\Phi_n}_{AB} = 
\frac{1}{n!\sqrt{n+1}}
(J^\dagger)^n
\ket{{\rm vac}}_{AB}.
\ee

Next, we express the creation operators in the diagonal polarization basis,
\be
a_H^\dagger =
\frac{a_D^\dagger+a_{\bar D}^\dagger}{\sqrt2},
\qquad
a_V^\dagger =
\frac{a_D^\dagger-a_{\bar D}^\dagger}{\sqrt2},
\ee
and similarly for Bob. Substituting these relations into $J^\dagger$, we obtain
\be
\begin{aligned}
J^\dagger
=
a_H^\dagger b_V^\dagger -
a_V^\dagger b_H^\dagger  \
=
a_{\bar D}^\dagger b_D^\dagger-
a_D^\dagger b_{\bar D}^\dagger.
\end{aligned}
\ee
Hence,
\be
(J^\dagger)^n =
\sum_{m=0}^{n}
(-1)^m
\binom{n}{m}
(a_{\bar D}^\dagger b_D^\dagger)^{n-m}
(a_D^\dagger b_{\bar D}^\dagger)^m.
\ee
Acting on the vacuum state gives
\be
(J^\dagger)^n
\ket{{\rm vac}}_{AB} = 
n!
\sum_{m=0}^{n}
(-1)^m
\ket{m,n-m}_{A_D}
\ket{n-m,m}_{B_D}.
\ee
Finally, by changing the summation index according to ($m'=n-m$), we find
\be
(J^\dagger)^n
\ket{{\rm vac}}_{AB} =
(-1)^n n!
\sum_{m'=0}^{n}
(-1)^{m'}
\ket{n-m',m'}_{A_D}
\ket{m',n-m'}_{B_D}.
\ee
Therefore, we have 
\be
\ket{\Phi_n}_{AB} =
\frac{(-1)^n}{\sqrt{n+1}}
\sum_{m=0}^{n}
(-1)^m
\ket{n-m,m}_{A_D}
\ket{m,n-m}_{B_D},
\ee
which has the same form as the expression in the $H/V$ basis.

\section{Closed-form expression for $\bar{Q}_Z$, $\bar{Q}_X$, $\bar{Q}_{Z_A, \bot_B}$ and $\bar{Q}_{\bot_A, Z_B}$} \label{appe: QW}
We first derive a closed-form expression for the expected value $\bar{Q}_W$ ($W\in\{Z,X\}$), and then derive one for $\bar{Q}_{Z_A, \bot_B}$. 
For simplicity, we only consider the case $r=1$, i.e., $\bar{p}_Z=p_Z$ and $\bar{p}_X=p_X$. 
For $W \in \{Z, X\}$,
we define the detection-related functions as
\begin{align}
y(n) &\coloneqq 1-(1-\eta)^n(1-d)^4, \\
y_W(n) &\coloneqq (1-d)^2\left((1-\eta(1-p_W))^n-(1-\eta)^n(1-d)^2 \right).
\end{align}
Here, $y(n)$ is a yield of an $n$-photon number pair for one party, i.e., a probability that $W_A$ (or $W_B$) is $Z$, $X$ or $\bot$ under the condition of $n$-photon-pair emittance. The yield $y_W(n)$ is a probability that $W_A$ (or $W_B$) is $W$ under the condition of $n$-photon emittance. 
In the expression of the yield $y_W(n)$, the factor $(1-d)^2$ indicates the probability that the dark count does not occur in the line of the other basis, the factor $(1- \eta (1 - p_W))^n$ is the probability that no photon goes into the line of the other basis, and the factor $(1- \eta )^n (1-d)^2 $ is the probability that all photons are lost and no dark count occurs in the $W$ line. Precisely, $y_W(n)$ is derived as follows. 
For $n\geq 1$, 
\be
\begin{split}
&y_W(n) \\
&= (1-d)^2 \left(  \sum\limits_{l=1}^{n} \binom{n}{l} (\eta p_W)^l (1-\eta)^{n-l} +(1-\eta)^n (1-(1-d)^2) \right) \\
&= (1-d)^2 \left(  \sum\limits_{l=0}^{n} \binom{n}{l} (\eta p_W)^l (1-\eta)^{n-l} -(1-\eta)^n (1-d)^2 \right) \\
& =  (1-d)^2 \left(  (\eta p_W + 1 - \eta)^n  -(1-\eta)^n (1-d)^2 \right) \\
& = (1-d)^2\left((1-\eta(1-p_W))^n-(1-\eta)^n(1-d)^2 \right).
\label{deriveyWn}
\end{split}
\ee
For $n = 0$,
\be
y_W(0) = (1-d)^2 (1-(1-d)^2),
\ee
which is obtained by substituting $n=0$ in the last line in Eq.~(\ref{deriveyWn}). 

Let us define 
\be
a_W \coloneqq 1-\eta (1-p_W), ~~b \coloneqq 1-\eta, ~~c \coloneqq (1-d)^2, 
\ee
for simplicity of notation. 
With the expressions, we rewrite
\begin{align}
y(n) &= 1-c^2 b^n,\\
y_W(n) &= c\big(a_W^n-c b^n\bigr).
\end{align}
The quantity $\bar{Q}_W$ is given by 
\begin{equation}
\begin{split}
\bar{Q}_W
&= \sum_{n=0}^{\infty}
P(n)
y_W(n)^2 \\
&= c^2
\sum_{n=0}^{\infty}
\frac{(n+1)\lambda^n}{(1+\lambda)^{n+2}}
\left(
a_W^{2n}
- 2c\, a_W^n b^n
+ c^2 b^{2n}
\right).
\label{eq:G_expand}
\end{split}
\end{equation}
Each term in Eq.~\eqref{eq:G_expand} is of the form
\[
\sum_{n=0}^{\infty}
\frac{(n+1)\lambda^n}{(1+\lambda)^{n+2}} t^n
=
\frac{1}{(1+\lambda)^2}
\sum_{n=0}^{\infty}
(n+1)
\left(
\frac{\lambda t}{1+\lambda}
\right)^n ,
\]
where $t \in \{a_W^2, a_Wb, b^2\}$.
By differentiating both sides of the relation
\begin{equation}
\sum_{n=0}^{\infty} x^n = \frac{1}{1-x}, \qquad |x|<1, 
\end{equation}
we obtain 
\begin{equation}
\sum_{n=0}^{\infty} (n+1)x^n
= \frac{1}{(1-x)^2},
\qquad |x|<1.
\label{eq:sum_identity}
\end{equation}
Since $|\lambda t/(1+\lambda)| <1$ holds for $t \in \{a_W^2, a_Wb, b^2\}$, 
we have 
\begin{equation}
\sum_{n=0}^{\infty}
\frac{(n+1)\lambda^n}{(1+\lambda)^{n+2}} t^n
=
\frac{1}{\left(1+\lambda-\lambda t\right)^2}.
\label{eq:key_sum}
\end{equation}
Applying Eq.~\eqref{eq:key_sum} to each term in Eq.~\eqref{eq:G_expand}, we find
\begin{align}
\bar{Q}_W
&= c^2
\Bigg[
\frac{1}{\left(1+\lambda-\lambda a_W^2\right)^2}
- \frac{2c}{\left(1+\lambda-\lambda a_W b\right)^2}
+ \frac{c^2}{\left(1+\lambda-\lambda b^2\right)^2}
\Bigg].
\end{align}


Next, the quantity $\bar{Q}_{Z_A, \bot_B}$ (and also $\bar{Q}_{\bot_A, Z_B}$) is given by 
\be
\begin{split}
\bar{Q}_{Z_A, \bot_B}
&= \sum_{n=0}^{\infty}
P(n) y_Z(n)\,\bigl[y(n)-y_Z(n)-y_X(n)\bigr]  \\
&=\sum_{n=0}^{\infty}
\frac{(n+1)\lambda^n}{(1+\lambda)^{n+2}}
\, y_Z(n)\,\bigl[y(n)-y_Z(n)-y_X(n)\bigr].
\end{split}
\ee
\\
\\
Since
\be
\begin{split}
y(n)-y_Z(n)-y_X(n)
&= \bigl(1-c^2 b^n\bigr)
 - c\bigl(a_Z^n-c b^n\bigr)
 - c\bigl(a_X^n-c b^n\bigr)  \\
&= 1 - c(a_Z^n+a_X^n) + c^2 b^n,
\end{split}
\ee
we have 
\begin{align}
&y_Z(n)\bigl[y(n)-y_Z(n)-y_X(n)\bigr]
= c\bigl(a_Z^n-c b^n\bigr)
   \Bigl[1-c(a_Z^n+a_X^n)+c^2 b^n\Bigr] \\
&= c\Bigl[
a_Z^n
- c a_Z^{2n}
- c a_Z^n a_X^n
+ c^2 a_Z^n b^n
- c b^n
+ c^2 b^n a_Z^n
+ c^2 b^n a_X^n
- c^3 b^{2n}
\Bigr].
\end{align}
Since $|\lambda t/(1+\lambda)| <1$ holds for $t \in \{a_Z, a_Z^2, a_Z a_X, a_Zb, b, a_X b, b^2\}$, 
applying Eq.~(\ref{eq:key_sum}) term by term, we obtain the closed-form expression
\begin{widetext}
\begin{align}
\bar{Q}_{Z_A, \bot_B}
&= c\Biggl[
\frac{1}{(1+\lambda-\lambda a_Z)^2}
- c\frac{1}{(1+\lambda-\lambda a_Z^2)^2}
- c\frac{1}{(1+\lambda-\lambda a_Z a_X)^2}
+ 2c^2\frac{1}{(1+\lambda-\lambda a_Z b)^2} \notag\\
&\qquad\qquad
- c\frac{1}{(1+\lambda-\lambda b)^2}
+ c^2\frac{1}{(1+\lambda-\lambda a_X b)^2}
- c^3\frac{1}{(1+\lambda-\lambda b^2)^2}
\Biggr].
\end{align}
\end{widetext}
Note that $Q_{\bot_A, Z_B}$ is equal to $Q_{Z_A, \bot_B}$ under the symmetric assumption of source-in-the-middle scenario.

\section{Closed-form expression for $\bar{e}_Z$ and $\bar{E}_X$} \label{appe:EW}
Here, we derive a closed-form expression for the expected values $\bar{e}_Z$ and $\bar{E}_X$. 
For simplicity, we only consider the case $r=1$, i.e., $\bar{p}_Z=p_Z$ and $\bar{p}_X=p_X$. 
Let $\bar{E}_Z$ denote the expected fraction of $Z$-basis error events among all rounds. 
Since the $Z$-basis error rate $e_Z$ defined in Eq.~(\ref{def:eZ}) is written as 
\be
e_Z=\frac{m_Z^{\rm sam,err}}{m_Z^{\rm sam}}=\frac{m_Z^{\rm sam,err}/m_{\rm rep}}{m_Z^{\rm sam}/m_{\rm rep}},
\ee
its expected value $\bar{e}_Z$ is 
\be
\bar{e}_Z=\frac{\bar{E}_Z(1-\gamma)}{\bar{Q}_Z(1-\gamma)/\gamma} = \frac{\bar{E}_Z\gamma}{\bar{Q}_Z},
\ee
where $\gamma$ is the non-sampling probability. 
We now derive the closed-form expression for $\bar{E}_Z$. 
First, we consider the error rate $E_{Z|nm}$, defined as the probability of a bit-error under the condition that 
\be
\ket{\phi_{nm}}_{AB} \coloneqq \ket{n-m,m}_A \ket{m,n-m}_B, 
\ee
which appears in Eq.~({\ref{Phi_n}), is emitted.  
When Alice receives at least one of the $n-m$ photons with polarization $H$ but none of the $m$ photons with polarization $V$ in the $Z$ line, while Bob receives at least one of the $n-m$ photons with polarization $V$ but none of the $m$ photons with polarization $H$, or with the roles of $H$ and $V$ interchanged for both Alice and Bob, the probability of these two cases is
\be
\begin{split}
p_{\rm same} \coloneqq & \left\{ (1-d)^2 [(1 -\eta(1-p_Z))^{n-m}-(1-\eta)^{n-m}](1-\eta)^m \right\}^2  \\
&+ \left\{ (1-d)^2 [(1 -\eta(1-p_Z))^{m}-(1-\eta)^{m}](1-\eta)^{n-m} \right\}^2 \\
= & C_1 ^2 + C_2 ^2,
\end{split}
\ee
where we defined 
\be
\begin{split}
C_1 &\coloneqq (1-d)^2 [(1 -\eta(1-p_Z))^{n-m}-(1-\eta)^{n-m}](1-\eta)^m \\
& = c (a_Z^{n-m} - b^{n-m})b^m, \\
C_2 &\coloneqq (1-d)^2 [(1 -\eta(1-p_Z))^{m}-(1-\eta)^{m}](1-\eta)^{n-m} \\
& = c (a_Z^m - b^m)b^{n-m}.
\end{split}
\ee
The probability of an error event arising from these two cases is
\be
p_{\rm same} e_d,
\ee
where $e_d$ denotes the error probability due to the detector imperfections and optical misalignment.
On the other hand, when Alice receives at least one of the $n-m$ photons with polarization $H$ but none of the $m$ photons with polarization $V$ in the $Z$ line, while Bob receives at least one of the $m$ photons with polarization $H$ but none of the $n-m$ photons with polarization $V$ in the $Z$ line, or with the roles of $H$ and $V$ interchanged for both Alice and Bob, the probability of these two cases is
\begin{widetext}
\be
\begin{split}
p_{\rm diff} \coloneqq & (1-d)^4 [(1 -\eta(1-p_Z))^{n-m}-(1-\eta)^{n-m}] [(1 -\eta(1-p_Z))^{m}-(1-\eta)^{m}] (1-\eta)^n  \\
&+(1-d)^4 [(1 -\eta(1-p_Z))^{m}-(1-\eta)^{m}] [(1 -\eta(1-p_Z))^{n-m}-(1-\eta)^{n-m}] (1-\eta)^n \\
= & 2 C_1 C_2.
\end{split}
\ee
\end{widetext}
The probability of an error event arising from these two cases is
\be
p_{\rm diff} (1-e_d).
\ee
For all other cases, the error probability is 1/2 due to double-click events or dark counts. Therefore, 
\be 
\begin{split}
E_{Z| nm} &= e_d p_{\rm same} + (1-e_d)  p_{\rm diff} + \frac{1}{2} (y_Z(n)^2 - p_{\rm same} - p_{\rm diff}) \\
&= \frac{1}{2} y_Z(n)^2 - \left(\frac{1}{2} - e_d\right) (p_{\rm same} - p_{\rm diff}) \\
&=  \frac{1}{2} y_Z(n)^2 - \left(\frac{1}{2} - e_d\right) (C_1^2 + C_2^2 - 2 C_1 C_2) \\
&=  \frac{1}{2} y_Z(n)^2 - \left(\frac{1}{2} - e_d\right) (C_1 - C_2)^2 \\
&= \frac{1}{2} y_Z(n)^2 - \left(\frac{1}{2} - e_d\right) c^2(a_Z^{n-m} b^m - a_Z^m b^{n-m})^2
\end{split}
\ee

Here, we define $E_{Z|n}$ as the probability of a bit in the $Z$ basis error under the condition that $n$-photon state Eq.~(\ref{Phi_n}) is emitted. 
This is written by 
\be
E_{Z|n} = \frac{1}{n+1} \sum_{m=0}^{n} E_{Z| nm}. \label{EZnm}
\ee
To derive the closed-form expression of $E_{Z|n}$, we calculate 
\begin{equation}
s(n) \coloneqq \sum_{m=0}^{n}
\left(
a_Z^{\,n-m} b^{\,m}
- a_Z^{\,m} b^{\,n-m}
\right)^2.
\end{equation}
To simplify each term, we rewrite
\begin{equation}
a_Z^{\,n-m} b^{\,m} = a_Z^n \left(\frac{b}{a_Z}\right)^m,
\qquad
a_Z^{\,m} b^{\,n-m} = b^n \left(\frac{a_Z}{b}\right)^m.
\end{equation}
Let
\begin{equation}
t_0 := \frac{b}{a_Z}~ (<1).
\end{equation}
Then we obtain
\begin{equation}
a_Z^{\,n-m} b^{\,m} - a_Z^{\,m} b^{\,n-m}
= a_Z^n t_0^m - b^n t_0^{-m}.
\end{equation}
Squaring this expression yields
\begin{equation}
\left(a_Z^{\,n-m} b^{\,m} - a_Z^{\,m} b^{\,n-m}\right)^2 = a_Z^{2n} t_0^{2m} - 2 a_Z^n b^n + b^{2n} t_0^{-2m}.
\end{equation}
Summing over $m = 0, \dots, n$, we obtain
\begin{equation}
s(n)
= a_Z^{2n} \sum_{m=0}^n t_0^{2m}
+ b^{2n} \sum_{m=0}^n t_0^{-2m}
- 2 a_Z^n b^n (n+1).
\end{equation}
Using the geometric series formulas (for $x \neq \pm 1$),
\begin{equation}
\sum_{m=0}^n x^{2m} = \frac{1 - x^{2(n+1)}}{1 - x^2},
\qquad
\sum_{m=0}^n x^{-2m} = \frac{1 - x^{-2(n+1)}}{1 - x^{-2}},
\end{equation}
we obtain 
\begin{equation}
\begin{split}
&a_Z^{2n} \sum_{m=0}^n t_0^{2m} + b^{2n} \sum_{m=0}^n t_0^{-2m} \\
&= a_Z^{2n}\frac{1-t_0^{2(n+1)}}{1-t_0^2} + b^{2n}\frac{1-t_0^{-2(n+1)}}{1-t_0^{-2}} \\
& = \frac{a_Z^{2n} -b^{2(n+1)}/a_Z^2}{1-b^2/a_Z^2} + \frac{b^{2n} -a_Z^{2(n+1)}/b^2}{1-a_Z^2/b^2}\\
&= \frac{a_Z^{2(n+1)} -b^{2(n+1)}}{a_Z^2 -b^2} + \frac{b^{2(n+1)} -a_Z^{2(n+1)}}{b^2 -a_Z^2} \\
& =  2\,\frac{a_Z^{2(n+1)} - b^{2(n+1)}}{a_Z^2 - b^2}. 
\end{split}
\end{equation}
Then $s(n)$ is described as 
\begin{equation}
s(n)
= 2\,\frac{a_Z^{2(n+1)} - b^{2(n+1)}}{a_Z^2 - b^2}
- 2(n+1)\,a_Z^n b^n.
\end{equation}
With this expression, Eq.~(\ref{EZnm}) leads to 
\be
E_{Z|n} =   \frac{1}{2} y_Z(n)^2 - \frac{1}{n+1} \left(\frac{1}{2} - e_d\right) c^2 s(n)
\ee

Then the overall QBER in the $Z$ basis is 
\be
\begin{split}
\bar{E}_Z &= \sum_{n=0}^{\infty} P(n)E_{Z|n} \\
&=\frac{1}{2} Q_Z - \left(\frac{1}{2} - e_d\right) c^2 \sum_{n=0}^{\infty} s(n) \frac{\lambda ^n}{(1+\lambda)^{n+2}} \\
&=\frac{1}{2} Q_Z - \left(\frac{1}{2} - e_d\right) \frac{c^2}{(1+\lambda)^2} \sum_{n=0}^{\infty} s(n) \kappa^n,
\label{EZformula}
\end{split}
\ee
where 
\be
\kappa := \frac{\lambda}{1+\lambda}. 
\ee
To evaluate the series
\begin{equation}
\sum_{n=0}^{\infty} s(n)\,\kappa^n,
\end{equation}
we decompose $s(n)$ as
\begin{equation}
s(n) = s_1(n) + s_2(n),
\end{equation}
with
\begin{equation}
s_1(n) = 2\,\frac{a_Z^{2(n+1)} - b^{2(n+1)}}{a_Z^2 - b^2},
\qquad
s_2(n) = -2(n+1)\,a_Z^n b^n.
\end{equation}
First term is calculated as 
\begin{align}
\sum_{n=0}^{\infty} s_1(n)\kappa^n
&= \frac{2}{a_Z^2 - b^2}
\sum_{n=0}^{\infty} \big(a_Z^{2(n+1)} - b^{2(n+1)}\big)\kappa^n \\
&= \frac{2}{a_Z^2 - b^2}
\left[
a_Z^2 \sum_{n=0}^{\infty} (\kappa a_Z^2)^n
- b^2 \sum_{n=0}^{\infty} (\kappa b^2)^n
\right].
\end{align}
Using the geometric series
\begin{equation}
\sum_{n=0}^{\infty} x^n = \frac{1}{1-x},~~|x|<1,
\end{equation}
we obtain
\begin{equation}
\sum_{n=0}^{\infty} s_1(n)\kappa^n
=
\frac{2}{a_Z^2 - b^2}
\left(
\frac{a_Z^2}{1-\kappa a_Z^2}
-
\frac{b^2}{1-\kappa b^2}
\right),
\end{equation}
because $|\kappa a_Z^2|<1$ and $|\kappa b^2|<1$ hold. 
Second term is calculated as 
\begin{align}
\sum_{n=0}^{\infty} s_2(n)\kappa^n
&= -2 \sum_{n=0}^{\infty} (n+1)(a_Z b)^n \kappa^n \\
&= -2 \sum_{n=0}^{\infty} (n+1)(\kappa a_Z b)^n.
\end{align}
Using
\begin{equation}
\sum_{n=0}^{\infty} (n+1) x^n = \frac{1}{(1-x)^2},~~|x|<1, 
\end{equation}
which is derived in Eq.~(\ref{eq:sum_identity}), we obtain
\begin{equation}
\sum_{n=0}^{\infty} s_2(n)\kappa^n
= -\frac{2}{(1-\kappa a_Z b)^2}
\end{equation}
because $|\kappa a_Z b|<1$ holds. 
Combining the two contributions, we obtain
\begin{equation}
\sum_{n=0}^{\infty} s(n)\kappa^n
=
\frac{2}{a_Z^2 - b^2}
\left(
\frac{a_Z^2}{1-\kappa a_Z^2}
-
\frac{b^2}{1-\kappa b^2}
\right)
-
\frac{2}{(1-\kappa a_Z b)^2}.
\end{equation}

Therefore, Eq.~(\ref{EZformula}) leads to 
\begin{widetext}
\begin{equation}
\bar{E}_Z
=
\frac{1}{2} Q_Z - \left(\frac{1}{2} - e_d\right) \frac{c^2}{(1+\lambda)^2}
\left[
\frac{2}{a_Z^2 - b^2}
\left(
\frac{a_Z^2}{1-\kappa a_Z^2}
-
\frac{b^2}{1-\kappa b^2}
\right)
-
\frac{2}{(1-\kappa a_Z b)^2}
\right].
\label{EZforumlafinal}
\end{equation}
\end{widetext}

From Eqs.~(\ref{Phi_n}) and (\ref{Phiequiv}), the state $\ket{\Phi_n}_{AB}$ has the same form in the $X$ basis as in the $Z$ basis. 
Accordingly, for $X$-basis, error fraction $E_X$ is obtained by taking $Q_Z \to Q_X$ and $a_Z \to a_X$ in Eq~(\ref{EZforumlafinal}).

\bibliography{kawakamibibD}
\bibliographystyle{apsrev4-1}

\end{document}